\documentclass[letter,11pt]{article}
\usepackage[letterpaper,top=1in,bottom=1in,left=1.4in,right=1.4in]{geometry}
\pdfoutput=1

\usepackage{lineno}

\usepackage{authblk}
\usepackage{orcidlink}

\usepackage[shortlabels]{enumitem} 

\usepackage{amsmath, amssymb, amsthm, mathtools}

\newtheorem{theorem}{Theorem}
\newtheorem{lemma}{Lemma}

\newtheorem{proposition}{Proposition}
\newtheorem{definition}{Definition}

\usepackage{bbm}
\usepackage{mathrsfs}
\usepackage{latexsym}
\usepackage{microtype}
\usepackage{xspace}

\PassOptionsToPackage{dvipsnames,usenames,table}{xcolor}
\usepackage{tikz}
\usetikzlibrary{calc,math}
\usetikzlibrary{arrows,shapes}

\newbox\mytempbox
\tikzstyle{embeds} = [->, >=open triangle 45]

\usepackage{xpunctuate}
\usepackage[all,british]{foreign} 
\redefnotforeign[ie]{i.e\xperiodafter} 
\redefnotforeign[eg]{e.g\xperiodafter}

\newenvironment{tightcenter}
 {\parskip=0pt\par\nopagebreak\centering}
 {\par\noindent\ignorespacesafterend}

\usepackage{marginnote} 

\usepackage[draft,english,silent]{fixme} 
\fxsetup{theme=color,layout=marginnote,innerlayout=noinline}

\usepackage{amsmath, amsthm, amssymb, amsfonts} 
\usepackage{thmtools} 
\usepackage{mathtools}

\usepackage[full,small]{complexity}

\usepackage{environ}

\usepackage{ifthen}
\newboolean{NotfullVersion}

\setboolean{NotfullVersion}{false} 

\NewCommandCopy\HungarianUmlaut\H
\renewcommand{\H}{\mathbb H}
\renewcommand{\G}{\mathbb G}
\newcommand{\J}{\mathbb J}
\renewcommand{\W}{\mathbb W}

\newcommand{\Naturals}{\mathbb{N}}
\newcommand{\Reals}{\mathbb{R}}
\newcommand{\Alg}{\cal{T}}

\newcommand{\defeq}{\coloneqq} 

\newcommand{\pp}{\operatorname{pp}}

\def\adm_#1{ \operatorname{adm}_{#1} }
\def\wcol_#1{ \operatorname{wcol}_{#1} }
\def\scol_#1{ \operatorname{scol}_{#1} }

\newtheorem*{theorem*}{Theorem}
\newlength{\RoundedBoxWidth}
\newsavebox{\GrayRoundedBox}
\newenvironment{GrayBox}[1]%
   {\setlength{\RoundedBoxWidth}{\textwidth-4.5ex}
    \def\boxheading{#1}
    \begin{lrbox}{\GrayRoundedBox}
       \begin{minipage}{\RoundedBoxWidth}%
   }{%
       \end{minipage}
    \end{lrbox}%
    \begin{tightcenter}%
    \begin{tikzpicture}%
       \node(Text)[draw=black!20,fill=white,rounded corners,%
             inner sep=2ex,text width=\RoundedBoxWidth]%
             {\usebox{\GrayRoundedBox}};
        \coordinate(x) at (current bounding box.north west);
        \node [draw=white,rectangle,inner sep=3pt,anchor=north west,fill=white] 
        at ($(x)+(6pt,.75em)$) {\boxheading};
    \end{tikzpicture}
    \end{tightcenter}\vspace{0pt}%
    \ignorespacesafterend
}

\usepackage[longend,vlined,plain,linesnumbered]{algorithm2e}

\SetKwInput{KwInput}{Input}
\SetKwInput{KwOutput}{Output}
\SetKwFor{RepTimes}{repeat}{times}{end}

\let\oldnl\nl
\newcommand{\nlnonumber}{\renewcommand{\nl}{\let\nl\oldnl}}

\newcommand{\Strata}{Strata\xspace}

\newcommand{\dStrata}[1]{#1\text{-}\Strata}

\newcommand{\length}{\mbox{length}}

\newcommand{\Nadir}{\operatorname{nadir}} 
\newtheorem{fact}[lemma]{Fact}

\newcommand{\GT}{\hat{\G}}

\def\Nesetril{Ne\v{s}et\v{r}il\xspace}

\title{A sufficient condition for characterizing the one-sided testable properties of families of graphs in the Random Neighbour Oracle Model}

\author[1]{Christine Awofeso\orcidlink{0009-0000-3550-1727}}
\author[1]{Patrick Greaves\orcidlink{0009-0007-0752-0526}}
\author[1]{Oded Lachish\orcidlink{0000-0001-5406-8121}}
\author[2]{Amit Levi\orcidlink{0000-0002-8530-5182}}
\author[1]{Felix Reidl\orcidlink{0000-0002-2354-3003}}

\affil[1]{%
\begin{minipage}[t]{.8\textwidth}%
\texttt{\{cawofe01|pgreav01\}@student.bbk.ac.uk}, \newline
\texttt{\{o.lachish|f.reidl@bbk.ac.uk\}@bbk.ac.uk}, \newline
Birkbeck, University of London, UK
\end{minipage}
}

\affil[2]{%
\begin{minipage}[t]{.8\textwidth}%
\texttt{alevi@cs.haifa.ac.il}, \newline University of Haifa, Haifa, Israel
\end{minipage}
}

\begin{document}

\maketitle
\begin{abstract}
We study property testing in the \emph{random neighbor oracle} model for graphs, originally introduced by Czumaj and Sohler [STOC 2019]. Specifically, we initiate the study of characterizing the graph families that are $H$-\emph{testable} in this model. A graph family $\mathcal{F}$ is $H$-testable if, for every graph $H$, $H$-\emph{freeness} (that is, not having a subgraph isomorphic to $H$) is testable with one-sided error on all inputs from $\mathcal{F}$.

Czumaj and Sohler showed that for any $H$-testable family of graphs $\mathcal{F}$, the family of testable properties of $\mathcal{F}$ has a known characterization, a major goal in the study of property testing.
Consequently, characterizing the collection of $H$-testable graph families will not only result in new characterizations, but will also exhaust this method of characterizing testable properties. We believe that our result is a substantial step towards this goal.

Czumaj and Sohler further showed that the family of planar graphs is $H$-testable, as is any family of minor-free graphs.
In this paper, we provide a sufficient and much broader criterion under which a family of graphs is $H$-testable.
As a corollary, we obtain new characterizations for many families of graphs including: families that are closed under taking topological minors or immersions, geometric intersection graphs of low-density objects, euclidean nearest-neighbour graphs with bounded clique number, graphs with bounded crossing number (per edge), graphs with bounded queue- and stack number, and more.

The criterion we provide is based on the \emph{$r$-admissibility} graph measure from the theory of sparse graph families initiated by \Nesetril and {Ossona de Mendez}.
Proving that specific families of graphs satisfy this criterion is an active area of research,
consequently, the implications of this paper may be strengthened in the future.
\end{abstract}
\newpage

\section{Introduction}\label{sec:Intro}
The framework of graph property testing deals with designing randomized algorithms, called \emph{testers} which, given query access to an input graph $G$, can decide whether the graph satisfies a given property $\mathcal{P}$ or is $\epsilon$-far from satisfying it---that is, a large part of the representation of the graph, whose size depends on $\epsilon$ and the graph representation, needs to be modified so that $G$ has the property. A central objective in property testing is to design testers whose query complexity is upper bounded by a function that is independent of the size of the input. A property is said to be \emph{testable} if it admits such a tester.

The seminal work of Goldreich, Goldwasser, and Ron~\cite{GGR98} initiated the study of property testing in the \emph{dense graph model}. In this model, the input graph is represented by its adjacency matrix, and the tester is allowed to query the entries in this matrix. Their work sparked a line of research that culminated in a precise characterization of testable properties in the dense graph model; for one-sided error testing see~\cite{alon2005every,alon2008characterization} and for two-sided error testing see~\cite{alon2006combinatorial}. 
These characterizations resolved two of the central open problems of this model.

Graph property testing is also studied in other models, including the \emph{general graph model}. In the general graph model~\cite{parnas2002testing, kaufman2004tight, alonTriangle2008}, the tester is given the number of vertices in the input graph as part of the input and is given access to an oracle that for every pair of vertices $u,v$ in the graph can answer the following queries:
(1) What is the degree $v$?
(2) What is the $i$th neighbour of $i$ (assuming one exists)?
(3) Are $u$ and $v$ adjacent?
It is also assumed implicitly that the tester knows the set of names of the graph vertices (for example,  it may $\{1,2,\dots,n\}$, where $n$ is the number of vertices in the graph.)

Many natural graph properties are \emph{not} testable in this model. Bipartiteness requires $\Omega(\min\{\sqrt{n}, n^2/m\})$ queries ~\cite{kaufman2004tight}, where $n$ is the number of vertices and $m$ is the number of edges and triangle-freeness requires $\Omega(n^{1/3})$ queries. This rules out the possibility of constant or even subpolynomial query testers for such basic properties. 
These lower bounds highlight the intrinsic limitations of testing in the general graph model. 

However, when the input graph is restricted to a family of graphs with additional `structure', many problems become testable: Czumaj \etal~\cite{czumaj2019planar} showed that when the input is restricted to planar graphs, bipartiteness is testable, Levi~\cite{TrianglesArboricity} showed that triangle-freeness is testable for graphs with bounded \emph{arboricity}, one of the most general sparse graph families which--among many other families--generalize planar graphs. 
On the other hand, $C_4$-freeness (i.e. $4$ cycle-freeness) is not testable for graphs of bounded arboricity (see~\cite{edenCycles2024}).

To obtain results for families that lie somewhere between these extremes, Awofeso \etal~\cite{H-freeness,AwofesoGLLR25} considered a hierarchy of sparse graph families by using the measure of \emph{$r$-admissibility}. We discuss this measure in more detail below; here, it is important to note that their result shows a sharp separation of testability across this hierarchy. Specifically, they show that $C_{2r}$- and~$C_{2r+1}$-freeness is testable when the input graph has bounded $r$-admissibility but requires $\Omega(n^{\alpha})$ queries, for some $\alpha > 0$, in graphs with bounded~$(r-1$)-admissibility.

A natural continuation to this line of research is to characterize the testable properties in the general graph model when the input is restricted to a family of graphs.
This goal has yet to be achieved. However, it turns out that if one considers a weaker model yet quite natural model, such results are possible.
Czumaj and Sohler~\cite{czumaj2019characterization}
Czumaj and Sohler~\cite{czumaj2019characterization}
introduced the \emph{random neighbour oracle} model, which is a variant of the general model in which the only query the tester can use is a \emph{random neighbour query}: the query is a vertex~$v$ in the graph, and the answer is a neighbour of~$v$ selected uniformly at random.
We note that while upper bounds in this model carry over to the general graph model, the same does not hold for lower bounds.

In this model, Czumaj and Sohler proved that every property that is testable with a one-sided error can be characterized by a finite set of forbidden subgraphs. 
This follows directly from the one-sidedness of a tester in this model and the fact that it can only ask random queries.
Specifically, by using the oracle the tester can discover whether two distinct vertices are adjacent but can only conclude with probability strictly lower than $1$, whether two distinct vertices are not adjacent. 
Now, since a one-sided error tester can only reject if the input graph is not in the family, it can only reject if it actually discovered a forbidden subgraph.

The relationship between testability and subgraphs is even stronger than that just described: Czumaj and Sohler noticed that the above implies that if $H$-freeness is testable with a one-sided error on a graph family~$\mathcal G$, then \emph{all} graph properties testable with a one-sided error on~$\mathcal G$ are exactly those that can be defined by a finite set of forbidden subgraphs. This allowed them to obtain a full characterization of the one-sided error testable properties for the family of planar graphs and any other family of minor-free graphs.

The random neighbour model and the above insights provide a powerful tool for characterizing testable properties for specific families of graphs: simply show that the family of graphs is $H$-\emph{testable}, which means that $H$-freeness is testable for every graph~$H$ with one-sided error. This sets a new and ambitious goal: characterizing all the $H$\emph{-testable} families of graphs, thus reaching the limits of this approach. In this paper, we take the first step in this direction by presenting a sufficient criterion under which a family of graphs is $H$-testable. We believe that our results represent a meaningful advance towards this broader characterization. To the best of our knowledge, previous work has focused only on specific graph families.


 The formulation of our criterion requires a collection of graph families denoted by $\mathcal{A}^p_r$, where $p, r \in \Naturals$. Each family $\mathcal{A}^p_r$ consists of all graphs whose $r$-admissibility is bounded by $p$. \emph{$r$-admissibility} is a widely-used measure from structural graph theory.
We provide an intuition for this measure later on in this section and a formal definition in Section~\ref{sec:Admissibility}.
For now it should be sufficient to note that
for every fixed $p\in \Naturals$ we have $\mathcal A^p_1 \supset \mathcal A^p_2 \supset \ldots \supset \mathcal A^p_\infty$, and that the result of Awofeso \etal~\cite{AwofesoGLLR25}, implies that \textbf{none of these graph families is $H$-testable}. However, our criterion for $H$-testability instead shows that suitable intersections of these families work. Specifically, we show that a family of graphs~$\mathcal F$ is $H$-testable if
there exists a function $f\colon\Naturals \longrightarrow \Naturals$, such that for every  
 $r\in \Naturals$ we have~$\mathcal F \subseteq \mathcal A^{f(r)}_r$. In other words, the $r$-admissibility of each member of the graph family is bounded by $f(r)$ for \emph{every} $r \in \Naturals$, and say that such a family has \emph{universally bounded $r$-admissibility}.  
 \marginnote{Universally bounded $r$-admissibility}

Thus, in contrast to the results on minor-closed families, $r$-admissibility provides both positive and negative results for $H$-testability: if for a family $\mathcal F$ there exists~$r,p \in \Naturals$ such that $\mathcal A^p_r \subseteq \mathcal F$, then $\mathcal F$ is not $H$-testable. If a family has universally bounded $r$-admissibility, then it is $H$-testable. This leads us to believe that our line of work has the potential to lead to a complete characterization of all $H$-testable graph families.

In order to provide some further context on how widely applicable our result is, we now present a (probably incomplete) list of graph families that are known to have universally bounded $r$-admissibility\footnote{
    Such families are known as having \emph{bounded expansion} in the literature, which is unrelated to the notion of expander graphs. For a consistent presentation, we prefer our terminology here.
} and hence are $H$-testable. 
From algorithmic and classical graph theory, we have graphs with a bounded tree-width, planar graphs, and graphs that exclude a (topological) minor or immersion~\cite{sparsity} (thus generalising Czumaj and Sohler's result). From graph drawing and colouring theory, we have graphs with bounded Euler genus, bounded crossing number per edge, bounded queue number, and bounded stack number (also known as bounded book thickness)~\cite{BndExpExamples}. 
Geometric graphs with low density also have universally bounded $r$-admissibility~\cite{harpeledApproximation2017}; notable examples include string graphs and intersection graphs of balls in $\mathbb{R}^d$ with constant depth, or related the nearest-neighbour graphs of points in $\mathbb{R}^d$ with bounded ply number~\cite{dvorak2022meta, miller1997separators}. Furthermore, many random graph models used for studying real-world networks generate graphs with universally bounded $r$-admissibility with high probability under suitable parameter regimes. These include Erdos--Rényi graphs with a bounded average degree~\cite{BndExpExamples}, the Chung--Lu model and the configuration model for degree sequences with sufficiently light tails~\cite{SparsityNetworks}, the stochastic block model~\cite{SparsityNetworks}, and the random intersection graph model~\cite{RIGs}.
The catalogue of graph families known to have universally bounded $r$-admissibility continues to grow and remains an active area of research.

This plurality of results stems from structural graph theory, initiated by the seminal work of Nešetřil and Ossona de Mendez~\cite{sparsity}, which aims to unify the treatment of diverse sparse graph families through a common framework. Central to this framework are measures such as $r$-admissibility, weak and strong $r$-colouring numbers, the $r$-neighbourhood complexity, and the density of $r$-shallow (topological) minors,.
A key observation is that if any one of these measures is bounded for some fixed $r$, then \emph{all of them are}\footnote{See \cite{sparsity} for the equivalence of various measures related to shallow (topological) minors, \cite{zhuColouring2009} for the equivalence of the weak- and strong colouring numbers to the shallow minor measures, \cite{reidlNeighbourhood2019} for the equivalence of neighbourhood complexity to shallow minor measures, and~\cite{dvorakDomset2013} for the equivalence of admissibility to the weak colouring number. We recommend the recent survey by Siebertz~\cite{siebertz2025survey} for a summary and detailed bibliography.}, usually for the same parameter~$r$ or within some small factor. Consequently, our result for \emph{$r$-admissibility} applies to all these structural parameters, thus the broad applicability of our criterion.

We formally define $r$-admissibility in Section~\ref{sec:Admissibility} and provide a short informal description now.
The family of all graphs with $1$-admissibility bounded by $p$ ($\mathcal A^p_{1}$), includes all graphs whose vertices can be arranged from right to left so that every vertex has at most $p$ neighbours on its left.
$\mathcal A^p_{2}$, includes all graphs whose vertices can be arranged from right to left so that the following holds:
every vertex $v$ can reach at most $p$ vertices on its left by a set of paths of length at most $2$, such that only $v$ is common to more than one of these paths, and for every one of these paths that has length $2$, the middle vertex is from the right of $v$.
$\mathcal A^p_{2}$, for $r>2$ is the natural generalization of the previous.

Note that these families of graphs are very general.
If one takes any graphs and replaces each one of its edges with a length path $2$ (which includes a new vertex), the result is a graph in $\mathcal A^2_{1}$. 
If a length $r+1$ paths were used instead of the length $2$ paths, the result is a graph in $\mathcal A^2_{r}$.

\paragraph*{Technical overview}

 We show that for a fixed graph~$H$ with $k$ vertices, $H$-freeness is testable under the random neighbour oracle model with one-sided error in the family~$\mathcal A^p_{k}$ where the query complexity depends only on~$k$, the proximity parameter and the bound~$p$ on the~$k$-admissibility. 

This work continues the study initiated by Awofeso \etal~\cite{AwofesoGLLR25}. Their analysis relied heavily on the structural properties of cycles. Moreover, they proved their results in the general graph model that, as stated, allows degree queries, which enabled them to easily treat vertices of low degree in a different manner from those of high degree. Such queries are not available in the model considered here, and our tester treats every vertex in the same manner.

As often with property testing, the testing algorithm itself is quite simple, and the challenge lies in proving the correctness of the algorithm.
In the following description of our~$H$-freeness tester, when we say `a small number', we mean a number that does not depend on the size of the graph.
Our tester for $H$-freeness initially selects a small random subset~$S$ of vertices of the input graph. Then, for each vertex in $S$ it uses a small number of neighbour queries. Once this is done, all the vertices returned by these queries are added to~$S$. This is repeated a small number of times. The tester rejects if the graph consisting of all the edges discovered this way has an $H$-subgraph, and otherwise it accepts. Next, we explain why this algorithm works as required.

If the input graph does not have an $H$-subgraph, then our algorithm does not discover an $H$-subgraph and accepts the input. 
In the analysis of the algorithm, when the input is $\epsilon$-far from $H$-freeness, we show that each of the following items holds with sufficiently high probability given the previous item holds (if there is one).
\begin{enumerate}[(i),nosep]
    \item One of the vertices in~$S$ is a vertex in a \emph{special} $H$-subgraph (we define special $H$-subgraphs later).
    \item Within a small number of iterations, the algorithm discovers a subset~$D$ of the vertices that we refer to as a \emph{stable} set for the special $H$-subgraph.
    \item Within a small number of iterations, a special vertex~$u$ referred to as a \emph{nadir} is discovered by the algorithm. For such a vertex it holds that $D\cup \{u\}$ is a stable set for some special $H$-subgraph, though not necessarily the original one. This implies that a eventually the algorithm discovers all the vertices of some $H$-subgraph.
    \item The algorithm discovers all the edges of the $H$-subgraph, whose vertices it discovered in the previous step.
\end{enumerate}

\noindent
The properties that make an $H$-subgraph of~$G$ `special' ensure that the surroundings of this subgraph are such that the above stages of the algorithm work with a sufficiently high probability. The algorithm itself does not need to know about this surrounding structure; in fact, it only appears in the analysis of our algorithm. To prove that special subgraphs exist, we use a trimming procedure in the spirit of~\cite{AwofesoGLLR25} that depends both on the structure of~$H$ as well as the structure of~$G$ and \emph{discovers} the existence of special subgraphs.

In summary, we obtain the following result using the above approach:

\begin{theorem}\label{thm:testable-F}
    Every graph family $\mathcal F$ with universally bounded $r$-admissibility is $H$-testable.
\end{theorem} 

\noindent
Combined with the result by Czumaj and Sohler~\cite{czumaj2019characterization} we therefore have a characterization of the testable properties of $\mathcal F$:
The theorem provides the necessary direction of the characterization, and Section~V of~\cite{czumaj2019characterization} provides the sufficient part\footnote{Note that these parts of~\cite{czumaj2019characterization} are general and are not specific to the planar or excluded-minor families}.

Theorem~\ref{thm:testable-F} is an immediate result of the following theorem, because if an input graph $G$ is from a universally bounded $r$-admissibility family, then~$G \in \mathcal A^p_r$ for $p= f(r)$.

\begin{theorem}\label{thm:testable-H}
    For every graph~$H$ and $p \in \Naturals$, $p \geq 2$, $H$-freeness is testable in $\mathcal A^p_r$ where~$r = |V(H)|$.
\end{theorem}

\noindent
Theorem~\ref{thm:testable-H} follows from the more technical Theorem~\ref{thm:main} in Section~\ref{sec:PBFS}.

\section{Preliminaries}

\marginnote{$\Naturals$,$\Reals$,$[k]$, $[k_1,k_2]$}
We write $\Naturals$ for the Natural numbers, $\Reals$ for the reals, and $\Reals^+$ for the set of positive reals.
For an integer $k$, we use $[k]$ as a shorthand for the set $\{ 1, 2, \ldots, k\}$.
For integers $k_1$ and $k_2$ such that $k_1 < k_2$, we use $[k_1,k_2]$ as a shorthand for the set $\{k_1,k_1+1 \ldots, k_2\}$ and $(k_1,k_2]$ as a shorthand for the set $\{k_1+1, k_1+2, \ldots, k\}$.
%

\newcommand{\dist}{\mbox{dist}}
All graphs considered in this work are simple undirected graphs.
For a graph $G$, we use $V(G)$ and $E(G)$ to refer to its vertex and edge set, respectively.  The degree of a vertex $v\in V(G)$, denoted $\deg_G(v)$, is the number of edges incident on $v$.
For a set $D\subseteq V(G)$, the notation $G[D]$ denotes the subgraph of $G$ induced on the vertices of $D$.
$N_G(v)$ is the set of neighbours of the vertex $v$ in the graph $G$.
The distance between two vertices $u,v\in V(G)$ is denoted $\dist_G(u,v)$. The distance between a vertex $u \in V(G)$ and a subset $D\subseteq V(G)$ is the minimum of $\dist_G(u,v)$ over every $v \in D$.

\marginnote{$xPy$}
For sequences of vertices~$x_1x_2\ldots x_\ell$ (and in particular, paths), we use the shorthand $x_1Px_\ell$ to represent the sequence. In addition, we use $Px_\ell$ (or $x_1P$) to refer to a subsequence of $x_1Px_\ell$. All paths considered in this work are simple paths. Though paths here are undirected, we often use terms used for directed paths by specifying a start and end vertex. 
The length of a path $P$ is the number of edges it has. We denote this value by $\length(P)$.

\marginnote{Ordered graph, -- subgraph}

An \emph{ordered graph} is a pair $\G = (G, \leq)$ where $G$ is a graph and $\leq$ is a
total order relation on $V(G)$.
We write $\leq_\G$ to denote the ordering of $\G$ and extend this notation to derive the relations $<_\G$,
$>_\G$, $\geq_\G$.
We use the notation we use for graphs also for ordered graphs; for example, $V(\G)$ and $E(\G)$ are $\G$'s vertex set and edge set, respectively. An \emph{ordered subgraph} of~$\G$ is an ordered graph~$\J = (J,\leq_\J)$ where~$J$ is a subgraph of~$G$ and~$\leq_\J$ is the order~$\leq_\J$ restricted to~$V(H)$.

\begin{definition} \marginnote{Order-isomorphism}
    Let $\J = (J_1, \leq_{\J_1})$ and $\J_2 = (J_2, \leq_{\J_2})$ be two ordered graphs.
    A graph isomorphism $\phi\colon V(\J_1)\longrightarrow V(\J_2)$ is an \emph{order-isomorphism} if for every two distinct vertices $u,v\in V(\J_1)$ we have $u <_{\J_1} v$ if and only if $\phi(u) <_{\J_2} \phi(v)$. 
\end{definition}

For an order-isomorphism $\phi\colon V(\J_1)\longrightarrow V(\J_2)$, and $D\subseteq V(\J_1)$, $W \subseteq V(\J_2)$ we write $\phi(D) = W$ to indicate that the image of~$D$ under~$\phi$ is~$W$. 

\marginnote{$H$-subgraph, $H$-subgraph}
In the following we will fix some graph~$H$ and an ordered graph~$\G$. To avoid lengthy descriptions, we will call an ordered subgraph~$\J$ of~$\G$ an \emph{$H$-subgraph} if $\J$ isomorphic to~$H$.


\paragraph*{Property Testing in the Random Oracle Graph Model}

\marginnote{Graph property, far, close}%
A \emph{graph property} (or simply \emph{property}) $\mathcal{P}$ is a family of graphs that is closed under isomorphism. A graph $G$ \emph{has the property $\mathcal{P}$} if $G\in \cal{P}$.

For a graph $G$ and the proximity parameter $0\leq \epsilon \leq 1$, we say $G$ is $\epsilon$-far from a property $\mathcal{P}$ if at least $\epsilon |V(G)|$ edges must be added to $G$ or removed from $G$ to obtain a graph in $\mathcal{P}$.
\marginnote{$H$-free, $\mathcal H$-free}
A graph $G$ is $\epsilon$-close to a property $\mathcal{P}$ if it is not $\epsilon$-far from it.

The central property in this work is $H$-freeness: for a fixed graph~$H$, a graph~$G$ is $H$-free if it does not contain a copy of~$H$ as a subgraph. 
For a family of graphs~$\mathcal H$, a graph~$G$ is $\mathcal H$-free if it does not contain any graph~$H \in \mathcal H$ as a subgraph. 

Note that when measuring the proximity to $H$-freeness, it is enough to consider the removal of edges. That is, a graph~$G$ is $\epsilon$-far from $H$-freeness if at least $\epsilon |V(G)|$ edges must be removed from~$G$ to get an $H$-free graph. This is simply due to the fact that the addition of edges cannot decrease the number of $H$-subgraphs in~$G$.

\marginnote{Random Oracle, i.u.r.}%
A \emph{random oracle} to a graph $G$, is an oracle that, given as input a vertex $v \in V(G)$, returns an independent and uniformly at random (i.u.r.) selected neighbour of $v$.

\begin{definition}\marginnote{Property tester}%
A \emph{property tester} for a property $\mathcal{P}$
of a family of graphs $\mathcal{G}$ is a randomized algorithm $\Alg$ that receives as input parameters $n\in \Naturals$, $\epsilon > 0$ and random oracle access to a graph $G = ([n],E)$ in $\mathcal{G}$.
If $G$ is $\epsilon$-far from $\cal{P}$, then $\Alg$ rejects with probability at least $2/3$. If the graph~$G\in \cal{P}$, then $\Alg$
accepts with probability~$1$ (if the tester has a \emph{one-sided} error)
or with probability at least~$2/3$ (if the tester has a \emph{two-sided} error).
\end{definition}

\noindent
In this paper, the tester has a one-sided error.
\marginnote{Query Complexity}%
The \emph{query complexity} of a property tester is the maximum number of queries it uses as a function of the distance parameter $\epsilon$, the size of the input $n$ and the attributes the property $\mathcal{P}$ and the family $\mathcal{G}$.
A property $\mathcal{P}$ is \emph{testable} for a family of graphs $\mathcal{G}$ if it has a property tester whose query complexity is independent of the size of the input graph.
We are particularly interested in graph families where $H$-freeness is testable for \emph{every} graph $H$, motivating the following terminology:

\begin{definition}\marginpar{$H$-testable}
    We say a graph family $\mathcal{F}$ is \emph{$H$-testable} if for every graph~$H$
    the property of being $H$-free is testable on $\mathcal{F}$.
\end{definition}


\section{Graph admissibility}\label{sec:Admissibility}

\newcommand{\Target}{\mathrm{Target}}

To define admissibility, we need the following concepts and notation.
\begin{definition}\label{def:pathAdmissible}\marginnote{$r$-admissible path}
Let $\G = (G,\leq)$ and $v \in V(G)$. 
A path~$vPx$ is $\G$-\emph{admissible} if $x <_\G v$ and $\min_{w\in P} {w} >_\G v$ (where the minimum is with respect to the order of the vertices of $\G$). That is, the path goes from~$v$ to~$x$ using only vertices $w$ such that $w >_\G v$ and~$x$ satisfies ~$v >_{\G} x$. An admissible path is 
\emph{$r$-admissible} if its length is at most $r$. 
\end{definition} 

\begin{definition}\marginnote{$\Target$}
For every integer $i > 0$, we let $\Target^i_{\G}(v)$ be the set of all vertices $u\in V(G)$ such that 
$u$ is reachable from $v$ via an \emph{$r$-admissible} path $vPu$ of length at most $i$. We omit the subscript~$\G$ when it is clear from context.  
\end{definition}

\begin{definition} \marginnote{$(r,\G)$-admissible path packing}
An \emph{$(r,\G)$-admissible path packing}
is a collection of paths~$\{vP_ix_i\}_i$ with a joint \emph{root}~$v$ and the additional properties that every path~$vP_ix_i$ is~$r$-admissible and the subpaths~$P_ix_i$ are all pairwise vertex-disjoint. In particular, all endpoints~$\{x_i\}_i$ are distinct. 
\marginnote{$\pp^r_\G$}
We write $\pp^r_\G(v)$ to denote the maximum size of any $r$-admissible path packing rooted at~$v$. 
\end{definition}

\noindent
Examples of $2$- and $3$-admissible path packings are depicted in the figure below.

\begin{definition}[Admissibility]\label{def:admissibility} \marginnote{$\adm_r(\G)$, $\adm_r(G)$}
The \emph{$r$-admissibility} of an ordered graph~$\G$, denoted $\adm_r(\G)$ and the $r$-admissibility of an unordered graph~$G$, denoted $\adm_r(G)$ are\looseness-1\footnote{%
    Note that some authors choose to define the admissibility as 
    $1 + \max_{v \in \G} \pp^r_\G(v)$ as this matches with some other, related measures.
} 
\[
\adm_r(\G) \defeq \max_{v \in \G} \pp^r_\G(v) \quad\text{and}\quad
\adm_r(G) \defeq \min_{\G \in \pi(G)} \adm_r(\G),
\]
where $\pi(G)$ is the set of all possible orders of the vertices of~$G$.
\end{definition}

\noindent
\marginnote{Admissibility ordering}%
If~$\le_\G$ is an order of the vertices of~$G$ such that $\adm_r(\G) = \adm_r(G)$, then we call $\le_\G$ an \emph{$r$-admissibility order} of~$G$. 

\begin{GrayBox}{}\small
    \begin{tightcenter}
    \includegraphics[width=0.95\textwidth]{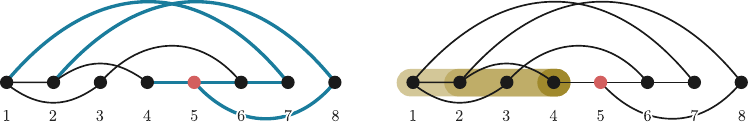}
    \end{tightcenter}
    \smallskip
    
    On the left, an example of a maximum $(3,\G)$-admissible path packing rooted at the vertex~$5$ in an ordered graph~$\G$. It contains the three paths $(5,4)$, $(5,6,7,2)$, and $(5,8,2)$. 
    The $3$-admissibility of this ordering is in fact three, with vertices~$5$ and~$6$ having $3$-admissible path packings of size three and all other vertices having packings of size less than three.
    
    On the right, the sets $\Target^1(5) = \{4\}$, $\Target^2(5) = \{2,3,4\}$, and $\Target^3(5) = \{1,2,3,4\}$
     are highlighted.
\end{GrayBox}

\begin{definition}
    \marginnote{$(p,r)$-admissible, $\mathcal A^p_r$}
\noindent
We say that a graph $G$ is \emph{$(p,r)$-admissible} if $\adm_r(G) \leq p$.
Note that, by definition, if a graph $G$ is $(p,r)$-admissible, it is also $(p,r')$-admissible for all~$r' \leq r$. We denote the family of all $(p,r)$-admissible graphs by~$\mathcal A^p_r$.
\end{definition}

\noindent
Note that $r$-admissibility is monotone under taking subgraphs, so if a graph is $(p,r)$-admissible, then so are all its subgraphs. We will frequently make use of this fact in the following way:

\begin{proposition}\label{prop:subgraph-order}
    Let~$G,J$ be graphs and let~$\G$ be a $(p,r)$-admissible ordering of~$G$. 
    Let $\J$ be a $H$-subgraph of~$G$ isomorphic to~$H$ and let~$\G'$ be the ordering of~$G'$ according to~$\G$.
    Then $\J$ is $(p,r)$-admissible.
\end{proposition}

\noindent
The second fact that we will frequently use is that if a graph is $(p,r)$-admissible, then the number of its edges is bounded from above:

\begin{fact} \label{fact:num-edges} 
If $G$ is $(p,r)$-admissible with $n$ vertices, then $|E(G)|\le pn $.
\end{fact}

\noindent
As indicated in the Introduction, we show that families whose $r$-admissibility can be bounded for every~$r$ are $H$-testable. Formally, such families are defined as follows:

\begin{definition}
    \marginnote{Universally bounded $r$-admissibility}
    A graph family $\mathcal F$ has \emph{universally bounded $r$-admissibility} if there
    exists a computable function~$f\colon \Naturals \to \Naturals$ such that for every~$r \geq 1$ it holds that $\mathcal F \subseteq \mathcal A^{f(r)}_r$.
\end{definition}

\noindent
The following is a well-known result in the field of sparse graphs; we reproduce it in Appendix~\ref{sec:appx1} using our notation for completeness. 

\newcommand{\PropSmallNeighbours}{}
\begin{restatable}{proposition}{ab}\label{prop:SmallNeighbours}
Let $\G = (G,\leq)$ be such that $\adm_r(G) = p$,
then for every $v\in V(G)$, and $h\in [r]$ we have $|\Target^h_{\G}(v)| \leq p(p-1)^{h-1}$.
\end{restatable}

\newcommand{\chain}{chain}
\begin{definition}\label{def:chain}
\marginnote{$\G$-\chain, $(r,\G)$-\chain}
Let $\G = (G,\leq)$.
A path $uLv$ is a $\G$-\emph{\chain} if for every $w\in uL$, we have $w>_\G v$. 
For $r\in \Naturals$, a $\G$-\chain~$uLv$ is an $(r,\G)$-\chain~if it does not have a $\G$-admissible subpath of length strictly larger than $r$.
\end{definition}

\noindent
We note that $(r,\G)$-\chain s differ from $(r,\G)$-admissible paths, since in an $(r,\G)$-chain $uLv$, there may exist a vertex $x\in L$, such that $u \geq_\G x \geq_\G v$, while if $uLv$ was an $(r,\G)$-admissible path, such a vertex does not exist. Consequently, an $(r,\G)$-chain decomposes into a sequence of $(r,\G)$-admissible paths.

The following propositions follow directly from the definitions of a \chain~and $(r,\G)$-admissible path (Definition~\ref{def:pathAdmissible}).
\begin{proposition}\label{prop:ChainVsAdmissible}
    Let $r \in \Naturals$ where $r>1$, and $\G = (G,\leq)$.
    Let $uLv$ be an $(r,\G)$-\chain~of length $t$ in $G$.
    There exists a \textbf{unique} set of $(r,\G)$-admissible paths $\{v_{i-1}L_i'v_i\}_{i=1}^{t'}$ such that $uLv = v_0L_1'v_1\dots v_{t'-1}L_{t'}'v_{t'}$.
\end{proposition}

\begin{proposition}\label{prop:subChain}
    Let $r \in \Naturals$ where $r>1$, and $\G = (G,\leq)$.
    If $uLv$ is an $(r,\G)$-\chain, or in particular an $(r,\G)$-admissible path, and $xL^*v$ be a subpath of $uLv$, then $xL^*v$ is an $(r,\G)$-\chain. 
    If $uLv$ is an $(r,\G)$-admissible path and $xL'u$ is a subpath of $uLv$, then $xL'u$ is an $(r,\G)$-\chain.
\end{proposition}

\noindent
We next prove a lemma for $(r,\G)$-chains  of the same nature as Proposition~\ref{prop:SmallNeighbours}. The lemma asserts that the size of every set of vertices reachable via $(r,\G)$-admissible paths from a fixed start vertex is bounded by a function of $\adm_r(\G)$ and~$r$. 

\begin{lemma}\label{lemma:chain-bundle}
    Let~$\G$ be an ordered graph with~$p := \adm_r(\G)$ and~$t \in \mathbb N$. Let $xP_1z_1,\ldots,xP_tz_t$ be a collection of $(r,\G)$-chains of length~$\leq \ell$,  where $z_i <_\G x$, for every $i\in [t]$.
    Then the set~$ \{ z_i \}_{i \in [t]}$  has size at
    most~$p^{r\ell}$.
\end{lemma}
\begin{proof}
    Let $V_{< x} := \{y \in V(G) \mid y <_\G x\}$, and $s = \max_{i\in [t]}|V(P_iz_i)\cap V_{< x}|$.
    That is, $s$ is the maximum number of vertices of $V_{< x}$ in a path~$P_iz_i$.     
    Next, we show that the set~$\{ \min_\G P_iz_i\}_{i\in [t]}$ has
    size at most~$p^{rs}$, since $s \leq \ell$ this implies the lemma.

    For~$s = 0$ the statement becomes vacuous, so consider~$s = 1$. This is the case where for every $i\in [t]$, the only vertex of~$P_iz_i$ smaller than~$x$ is $z_i$. 
    Since all paths $P_iz_i$ are $(r,\G)$-\chain s,  we have~$ \{ z_i \}_{i \in [t]} \subseteq \Target^r_\G(x)$. Therefore, by Proposition~\ref{prop:SmallNeighbours}, $|\{\min_\G P_i\}_{i=1}^{t}| \leq p (p - 1)^{r-1} < p^r$, which provides us with the inductive base.

    Assume that the statement holds for all~$s' < s$, where~$s \leq \ell$, and all collections $Q_1,\dots,Q_h$ of $(r,\G)$-chains of length~$\leq \ell$, such that for every $i\in [h]$, $Q_i$ starts at the vertex~$z \in \G$ and ends at some vertex smaller than~$z$ under~$\leq_\G$, and $ |V(Q_i)\cap V_{< z}| \leq s'$.

    Let $\mathcal P = \{P_i\}_{i=1}^{t}$ and for every $i\in [t]$, let~$y_i \in V(P_i)$ be such that $y_i <_\G x$ and $P_i=xL_iy_iR_iz_i$, where $|V(xL_iy_i) \cap V_{< y_i}| = 1$.
    According to the base of the induction $|\{y_i\}_{i=1}^{t}|\leq p^{r}$.
  
    Pick a vertex~$y \in \{y_i\}_{i=1}^{t}$ and let $\mathcal P_y$ be the collection of all the paths $y_iR_iz_i$ such that $y_i = y$.
    Since $y <_\G x$, every path in $\mathcal P_y$ is a subpath of a path in $\mathcal P$, we can conclude that for every $yRz\in \mathcal P_y$, we have $|V(yRz)\cap V_{<y}| < s$.
    Thus, by the induction assumption, $|\{ \min_\G P\}_{P\in \mathcal P_y}| \leq p^{r(s-1)}$.
    Finally, by construction, we have $|\{ \min_\G P\}_{P\in \mathcal P}| \leq p^{r}\cdot |\{ \min_\G P\}_{P\in \mathcal P_y}| \leq p^{sr} \leq p^{r\ell}$.   
\end{proof}

\newcommand{\Stratas}{Stratas\xspace}
\newcommand{\Hvy}{\operatorname{Heavy_\alpha(G)}}
\newcommand{\Crt}{\operatorname{Critical}}

\section{Notations and Overview}

In order to simplify the notation in the sequel, we fix $r, p \in \Naturals$, $G$ and $H$ to be $(p,r)$-admissible graphs. Note that if $H$ was not $(p,r)$-admissible, then it could not be a subgraph of~$G$ as the admissibility does not increase when taking subgraphs. We further assume $H$ to be connected and that $|V(G)| > 2r \geq |V(H)|$. 

We continue by first providing a short informal description of our algorithm.
We proceed with an overview of our analysis of the algorithm and the required definitions.

Given access to a random-neighbour oracle to the input graph $G$, our algorithm starts by selecting a small random subset $S \subseteq V(G)$. For each vertex $v \in S$, it queries a fixed number of random neighbours $u$ (thus discovering the edges $uv$). All discovered vertices are added to $S$. This “growth” process is repeated a fixed number of times, producing a small subgraph of $G$ whose edges are exactly those uncovered during exploration. The algorithm rejects if this subgraph contains an $H$-subgraph, and accepts otherwise.
The exact number of queries used by the algorithm depends only on $\epsilon$, $p$, and $r$.

We note that because of its rejection condition, our algorithm will never reject an $H$-free input graph.
Hence, the overall goal is to show why our algorithm rejects with sufficient probability given that the input graph $G$ is $\epsilon$-far from $H$-freeness.
So, from now on in this section, we assume that $G$ is $\epsilon$ far from $H$-freeness.
Also, from now on in this section, we let $\G$ be a $(p,r)$-admissible order of~$G$, 

The analysis of our algorithm is based on a process that constructs a subgraph $\GT$ of $\G$,
which we show to be $\epsilon/2$-far from being $H$-free and refer to as the \emph{trimmed} graph. 
It is important to note that $\GT$ includes all vertices of $\G$
and that $\GT$ has the same order as $\G$.
We formally show how $\GT$ is obtained in Section~\ref{sec:Trimming}.
In this section, we mention more attributes of $\GT$ when they are required.
Note that the algorithm does not use or compute any information on ordered graphs $\G$, $\GT$.

Our proof reduces to showing that each of the following events occurs with sufficiently high probability conditioned on the prior event happening with sufficiently high probability. 

\begin{enumerate}
    \item\label{step:hit} A vertex that lies in an $H$-subgraph $\J$ of $\GT$ is in $S$.
    \item\label{step:dig} The vertex $\min_\J V(\J)$ is discovered.
    \item\label{step:cover} The vertices of some $H$-subgraph $\J^*$ of $\GT$, that is not necessarily $\J$, are discovered one by one.
    \item\label{step:shore} All the edges of $\J^*$ are discovered. 
\end{enumerate}

\noindent
We proceed by explaining Step~(\ref{step:admissiblePaths}) and (\ref{step:shore}) together, then we explain Step~(\ref{step:dig}) and finally we explain Step~(\ref{step:cover}) which is the crucial part of our analysis and requires most of the effort.

\paragraph{(\ref{step:hit}) and (\ref{step:shore})}
When $\GT$ is constructed from $\G$, we ensure that every vertex $v$ which has a high enough degree in~$G$ has no neighbour which precedes it in the graph~$\GT$, that is, for every~$u \in N_{\GT}(v)$ it holds that~$u >_\G v$.
This means that for every edge~$uv$ in $\GT$ either $u$ or $v$ has a \emph{small} degree in $G$.
Because $H$ is connected and we can assume that it has more than one vertex, every $H$-subgraph of $\GT$ contains a vertex which has a small degree in $G$. 
Since $\GT$ will by construction be $\epsilon/2$-far from $H$-freeness, a constant portion of the vertices of $\G$ are in some $H$-subgraph of $\GT$. 
Hence, if $S$ is sufficiently large, then it is highly likely that it includes such a vertex. 

Let us now briefly argue why Step~(\ref{step:shore}) works. Assume that the algorithm has already discovered all the vertices of some $H$-subgraph of $\GT$ and consider some edge~$uv$ that
belongs to this subgraph. Since one of $u,v$ has a small degree, it is highly likely that the edge~$uv$ is discovered by the algorithm, and since this holds for every edge of the subgraph, all edges are discovered with a sufficiently high probability.

\paragraph{(\ref{step:dig})}
When we construct $\GT$ from $\G$, we ensure that for every vertex $x$ that has an $(r,\G)$-admissible path $xPy$ in $\GT$ there exist sufficiently many paths $xP'y$ in $\GT$.
Specifically, the number of such edge-disjoint paths will be a constant portion of the degree of $x$. 
This ensures that with high probability, if $x$ is in $S$, then $y$ will be added to $S$ with high probability. Let us now see how this property ensures that Step~\ref{step:dig} works with a high enough probability.

Assume that~$\J$ is an $H$-subgraph of~$\GT$ and that, by the previous step, the set~$S$ contains some vertex
$v \in \J$. We argue that then $x_1 = \min_\J V(\J)$ is added to~$S$ with high probability. Since $v$ is in $\J$ and $|V(\J)| \leq r$, we are guaranteed that there exists an $(r,\G)$-chain $vPx_1$ in $\GT$. Since~$\GT$ is a subgraph of~$\G$ with the same order, this chain also exists in $\G$. By Proposition~\ref{prop:ChainVsAdmissible}, we are therefore guaranteed that with high probability, $x_1$ will be added to $S$.

\paragraph{(\ref{step:cover})}
Assume by the previous step that~$S$ includes the vertex $x_1 = \min_\J V(\J)$ of some $H$-subgraph $\J$ of~$\GT$. We show that then with high probability $S$ will include all the vertices of an $H$-subgraph that includes $x_1$, though that subgraph might not be~$\J$.
We start by explaining why for some $H$-subgraph $\J'$, $S$ will include the vertices
$\min_\G V(\J^*)$ and $\min_\G V(\J^*)\setminus \{x_1\}$.

When we construct $\GT$ from $\G$, we ensure that for every vertex $y$ with $y= \min_{\J^*}V(\J^*)$ for some $H$- subgraph $\J^*$ there exists a sufficiently large number of edge-disjoint $H$-subgraphs $\J'$ such that $y = \min_{\J'}V(\J')$. Here, sufficiently large means $\Omega(\deg_G(y))$.

Let $x_2^1,\dots,x_2^t$, which are the second smallest in one of the edge-disjoint $H$-subgraphs above.
By definition, for every $i\in [t]$, there exists an $(r,\GT)$-admissible path $x_2^iP_iy_ix_1$ (if the length is $1$, then the path is simple $x_2^i$).
By Proposition~\ref{prop:subChain}, for every $i\in [t]$, the subpath $y_i{P_i}'x_2^i$ of $x_2^iP_iy_ix_1$ is a $(r,\GT)$-chain.
So, if one of the vertices $y_1,\dots,y_t$ is in $S$, 
by the same reasoning as in Step~(\ref{step:dig}), with a high probability one of the vertices $x_2^1,\dots,x_2^t$ will be in $S$.
Now, since the value of $t$ is $\Omega(\deg_G(x_1))$, the probability of this is very high.

The above already looks promising, but it is not quite sufficient as there are some subtle complications which require additional machinery. Consider the following case: Let $x_1\dots,x_{|V(H)|}$ be the vertices of the $H$-subgraph $\J$, numbered as they appear in the ordering of~$\J$.
Suppose that~$S$ already includes the vertices $x_1,\dots,x_{s-1}$.
Then it may be the case that $x_s$ is in a connected induced subgraph $\W$ of $\J$ such that $V(\W) \cap \{x_1,\dots,x_{s-1}\} = \{x_1,x_2\}$ and $\{x_1,x_2\}$ is a vertex-separator in $\J$.
Now, if we find some ${x_s}' \neq x_s$ with the same role in some $H$-subgraph $\J'$ of $\GT$, which includes $\{x_1,x_2\}$, we have no guarantee that $\J'$ includes all vertices $x_3,\dots,x_{s-1},x_{s+1},\dots,x_{|V(H)|}$.

To resolve this issue, we ensure that ${x_s}'$ and $\J'$ are such that the induced subgraph $\W'$ of $\J'$ that is mapped to $\W$ by an order isomorphism from $\J'$ to $\J$ is such that $V(\W')\cap \{x_1,\dots,x_{s-1},x_s,x_{s+1},\dots,x_{|V(H)|}\} = \{x_1,x_2\}$.
Such an order isomorphism exists because $\J$ and $\J'$ are both $H$-subgraphs.
We note that the subgraph we get by taking the vertices
$V(\W') \cup (V(\J)\setminus V(\W'))$, the edges of $\W'$ and the edges of $\J$ induced on $\{x_1,x_2\} \cup (V(\J)\setminus V(\W'))$ is an $H$-subgraph of $\GT$.

For our construction, the fact that the vertices $x_1,x_2$ are smaller according to the $\W'$ of any other vertex in $\W'$ is crucial for the above to work.
The importance of such a set of vertices is captured by the following definition.

\begin{definition}\label{def:prefix} \marginnote{Prefix}
    Let~$\J$ be an $H$-subgraph of~$\G$ and let~$\G'$ be an ordered subgraph of~$\J$.
    A \emph{prefix} of $\G'$ is a \textbf{non-empty} subset $D\subseteq V(\G')$ that satisfies the following:
    \begin{enumerate}
        \item for every $u\in D$ there exists $v\in V(\G') \setminus D$ such that $uv\in E(\G')$,
        \item\label{item:Prefix2} for every $u\in D$ and $v\in V(\G') \setminus D$, we have $v>_{\G'} u$, and
        \item if $v \in V(\G')$ has a neighbour in $V(\J)\setminus V(\G')$, then $v\in D$.
    \end{enumerate}
\end{definition}
\noindent 
Note that if $\G' \neq \J$ in the above definition, then a prefix is a separator in~$\J$.
We always use a prefix with its ordered subgraph and hence the following definition.

\begin{definition}\label{def:Huseful} \marginnote{Useful pair}
    Let~$\J$ be an $H$-subgraph of~$\G$ and let~$\G'$ be an induced ordered subgraph of~$\J$.
    Let~$U_{\G'}$ be a prefix of~$\G'$. Then~$(\G', U_{\G'})$ is a \emph{useful pair}
    if~$\J[V(\G')\setminus U_{\G'}]$ is connected.
\end{definition}

\noindent
In the above overview, we actually replace one useful pair with another during the progressive discovery. In order for this to work it need to satisfy the following condition:

\begin{definition}\label{def:HSimilar} \marginnote{Similar}
    Two useful pairs $(\G_1,U_{\G_1})$ and $(\G_2,U_{\G_2})$ are \emph{similar} 
    if $U_{\G_1} = U_{\G_2}$ and there exists an order-isomorphism $\phi\colon V(\G_1) \longrightarrow V(\G_2)$ such that~$\phi(U_{\G_1}) = U_{\G_2}$.
\end{definition}

\begin{GrayBox}{}\small
    \begin{tightcenter}
    \includegraphics[width=0.95\textwidth]{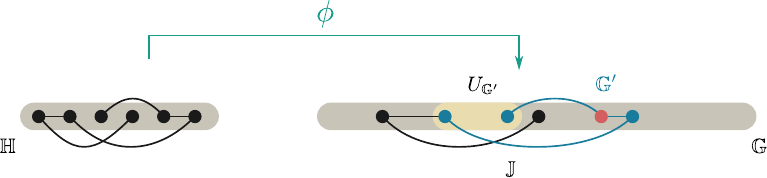}    
    \end{tightcenter}
    An $H$-subgraph $\J$ with order-isomorphism $\phi$. Inside of~$\J$, a subgraph $\G'$
    with prefix~$U_{\G'}$. Since~$\J[V(\G') \setminus U_{\G'}]$ is connected, the tuple $(\G', U_{\G'})$ is a useful pair. The red vertex in~$\G'$ is the \emph{nadir} of $(\G', U_{\G'})$ defined further below.
\end{GrayBox}
The following is a simple proposition that we need need further on.

\begin{proposition}\label{prop:NoEmptyUseful}
    Let $(\G', U_{\G'})$ be a useful pair.
    If $\G'$ is not the empty graph, then $U_{\G'} \subsetneq V(\G')$.
\end{proposition}
\begin{proof}
    By definition, the prefix $U_{\G'}$ is not empty and hence $\G'$ is not the empty graph.
    Let $u \in U_{\G'}$.
    By the definition of a prefix, there exists $v\in V(\G') \setminus U_{\G'}$ such that $uv\in E(\G')$, and hence $U_{\G'} \subsetneq V(\G')$.
\end{proof}

\newcommand{\MD}{\mathcal M}
\noindent
To prove that the vertices of some $H$-subgraph of $\GT$ are found one by one, we need a plurality of similar useful pairs all sharing the same prefix. We formalize this notion as follows:

\begin{definition}\label{def:strata} \marginnote{Strata}
The tuple $(\MD,U)$ is a \emph{\Strata} if  $\MD$ is a family of pairwise similar useful pairs with the following properties:
\begin{enumerate}
    \item For every $(\G^*,U^*)\in \MD$, we have $U^*=U$, and
    \item for every distinct $(\G_1^*,U),(\G_2^*,U)\in \MD$, the only edges shared by $\G_1$ and $\G_2$ are the ones between the vertices of $U$.
\end{enumerate}
\end{definition}

\begin{GrayBox}{}\small
    \begin{tightcenter}
    \includegraphics[width=0.65\textwidth]{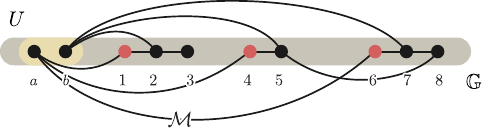}    
    \end{tightcenter}
    \smallskip

    An example strata $(\mathcal M, \{a,b\}$, where $\mathcal M$
    contains the useful pairs $(\G[\{a,b, 1,2,3\}], \{a,b\})$, $(\G[\{a,b, 4,5,8\}], \{a,b\})$,
    and $(\G[\{a,b,6,7,8\}], \{a,b\})$ which are all pairwise similar. The red vertices are the nadirs of each useful pair, as defined further below. Note that the useful pairs are edge-disjoint but not necessarily vertex-disjoint; here, vertex 8 is shared between two members.
\end{GrayBox}


\noindent
In the analysis, we show that if our algorithm discovered all the vertices in a set $U$ of a \Strata~$(\MD,U)$, then, with high probability, it will discover a special vertex in some in an ordered subgraph $\G'$, where $(\G',U)$ is some member of $\MD$.
The following definition captures its specific attributes, which we described informally in Step~(\ref{step:cover}) above.

\begin{definition}\label{def:Nadir}\marginnote{Nadir}
 The $\Nadir$ of a useful pair $(\G', U_{\G'})$, denoted $\Nadir(\G', U_{\G'})$, is the vertex $v\in V(\G')\setminus U_{\G'}$ such that for every $u\in V(\G')\setminus U_{\G'}$ we have $v \leq_\G u$.
 The $\Nadir$ of a \Strata~$(\MD,U)$, denoted $\Nadir(\MD,U)$ is the set $\{\Nadir(\G^*, U)\}_{(\G^*, U)\in \MD}$.
\end{definition}

\noindent
The next proposition follows directly from the above definition and Proposition~\ref{prop:NoEmptyUseful} which states that if a useful pair is not empty, then its prefix does not include all the vertices of the useful pair.
\begin{proposition}\label{prop:Nadir}
    For every non-empty useful pair $(\G', U)$, $\Nadir(\G', U)$ exists.
\end{proposition}

\noindent
The existence of sufficiently large \Stratas~in $\GT$ (when the input graph is far from being $H$-free) is ensured by the trimming.
Now there can be two cases for a \Strata~$(\MD,U)$, 
a vertex like $x_s$, described in Step~(\ref{step:cover}) appears as a nadir in a significant number of useful pairs in $(\MD,U)$ or none of the vertices in $\Nadir(\MD,U)$ appear as a nadir in a significant number of useful pairs in $(\MD,U)$.
In the first case we show that, with sufficiently high probability, the algorithm discovers this vertex. In the second case we show that, with a sufficiently high probability, the algorithm finds a replacement as described in Step~(\ref{step:cover}).
The following two definitions are used to capture the above.

\begin{definition}\marginnote{$\delta$-weak, $\delta$-strong}
    Let $\delta \in \Naturals$ and $(\MD,U)$ be a \Strata. A vertex $v\in \Nadir(\MD,U)$ is \emph{$\delta$-weak} for $(\MD,U)$
    if the maximum size of a \Strata~$(\MD',U)$, such that $\Nadir(\MD',U) = \{v\}$, is at most $\deg_{G}(\max_\G U)/\delta$.
    Otherwise we call~$v$ \emph{$\delta$-strong} for $(\MD,U)$.
\end{definition}

\begin{definition}\marginnote{\dStrata{$\delta$}}
    For $\delta \in \Naturals$ 
    a \Strata~$(\MD,U)$ is a \emph{\dStrata{$\delta$}} if every $v\in \Nadir(\MD,U)$ is $\delta$-weak.
\end{definition}

\begin{GrayBox}{}\small
    \begin{tightcenter}
    \includegraphics[width=0.65\textwidth]{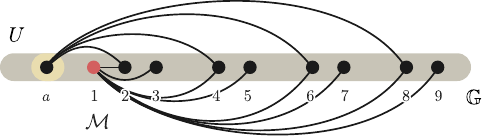}    
    \end{tightcenter}
    \smallskip

    An example strata $(\mathcal M, \{a\})$, where $\mathcal M$
    contains the useful pairs $(\G[\{a,1,2,3\}], \{a\})$, 
    $(\G[\{a,1,4,5\}], \{a\})$,
    $(\G[\{a,1,6,7\}], \{a\})$, and
    $(\G[\{a,1,8,9\}], \{a\})$.
    All useful pairs share the red vertex~1 as their nadir.

    If we assume that $\mathcal M$ is maximum for this nadir and that \eg $\deg_G(a) = 16$, then the nadir~1 is $4$-weak since $|\mathcal M| \leq \deg_G(a) / 4$, but it would be $5$-strong since $|\mathcal M| > \deg_G(a) / 5$.
\end{GrayBox}

\noindent
To prove that the algorithm discovers the vertices of some $H$-subgraph of $\GT$, we need  to reason about the useful subgraphs `attached' to some subset of vertices, motivating the following definition:

\newcommand{\MSpine}{\mbox{Spine}}
\newcommand{\Spine}{Spine}
\begin{definition}\label{def:Spine}\marginnote{\Spine}
    Let~$\J$ be an $H$-subgraph of~$\G$ and~$U \subseteq V(\J)$.
    The \emph{\Spine}~of a subset $U$ of~$V(\J)$, denoted $\MSpine_\J(U)$,
    is the family of all useful pairs $\{(\G_i,U_i)\}_{i=1}^s$ such that for every $i \in [s]$, $\G_i$ is an ordered subgraph of~$\J$ and $V(\G_i)\cap U = U_i$.
\end{definition}

\begin{lemma}\label{lem:Spine}
    Let~$\J$ be an~$H$-subgraph in~$\G$. Let $U\subseteq V(\J)$ and $\MSpine_\J(U) = \{(\G_i,U_i)\}_{i=1}^s$ and let $i$ and $j$ be distinct elements of $[s]$.
     Then no edge of~$\J$ has one endpoint in $V(\G_i)\setminus U_i$ and the other in $V(\G_j)\setminus U_j$. Furthermore, $(V(\G_i)\setminus U_i) \cap (V(\G_j)\setminus U_j) = \emptyset$.
\end{lemma}
\begin{proof}
    Let $i$ and $j$ be distinct elements of $[s]$.
    Suppose for the sake of contradiction that $V(\G_i) = V(\G_j)$.
    Since $U_i = V(\G_i)\cap U = V(G_j)\cap U = U_j$ and, by definition, both $\G_i$ and $\G_j$ are induced ordered subgraphs of $\J$, we have $(\G_i,U_i) = (\G_j,U_i)$ in contradiction to $(\G_i,U_i) = (\G_j,U_i)$ being distinct members of $\MSpine_\J(U)$.
    Thus, $V(\G_i) \neq V(\G_j)$.

    Assume for the sake of contradiction that $V(\G_i)\setminus U_i = V(\G_j)\setminus U_j$.
    Since $V(\G_i) \neq V(\G_j)$, we may assume without loss of generality that there exists a vertex $u \in U_i\setminus U_j$ and otherwise we rename the sets accordingly.
    By the definition of a useful pair, this means that $u$ is adjacent to a vertex in $V(\G_i)\setminus U_i$, thus also to a vertex in $w\in V(\G_j)\setminus U_j$.
    By the definition of a useful pair, this means that $w\in U_j$ which contradicts the above.
    So, $V(\G_i)\setminus U_i \neq V(\G_j)\setminus U_j$.

    There cannot exist an edge $uv$ such that $u\in (V(\G_i)\setminus U_i)$ and possibly in $(V(\G_j)\setminus U_j)$ and $v$ is only in $(V(\G_j)\setminus U_j)$, since this would imply that $u \in U_i$ in contradiction to $u\in (V(\G_i)\setminus U_i)$.
    In the same manner there cannot exist an edge $uv$ such that $u\in (V(\G_j)\setminus U_j)$ and possibly in $(V(\G_i)\setminus U_i)$ and $v$ is only in $(V(\G_i)\setminus U_i)$.
    So, there no edge of~$\J$ has one endpoint in $V(\G_i)\setminus U_i$ and the other in $V(\G_j)\setminus U_j$.
    
    Finally, we notice that it cannot be the case that $(V(\G_i)\setminus U_i) \cap (V(\G_j)\setminus U_j) \neq \emptyset$, since by definition both the induced subgraph of $\J$ on $V(\G_i)\setminus U_i$ and the induced subgraph of $\J$ on $V(\G_j)\setminus U_j$ are connected, so this would imply the existence of an edge of the sort we just proved does not exist.
\end{proof}

\begin{GrayBox}{}\small
    \begin{tightcenter}
    \includegraphics[width=0.8\textwidth]{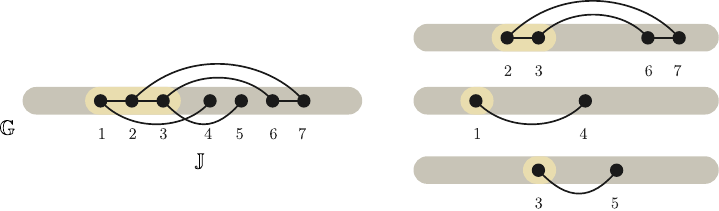}    
    \end{tightcenter}

    An example $H$-subgraph $\J$ on the left. On the right, the members $(\G[\{2,3,6,7\}], \{2,3\})$, 
    $(\G[\{1,4\}], \{1\})$, and
    $(\G[\{3,5\}], \{3\})$
    of $\MSpine_\J(\{1,2,3\})$.
\end{GrayBox}

\noindent
Finally, the following definition and its accompanying lemma enable us to prove that indeed some of the vertices that our algorithm discovers are all contained in some $H$-subgraph.

\begin{definition}\label{def:Hstable}\marginnote{$\J$-stable}
\sloppy
    Let~$\J$ be an $H$-subgraph in~$\G$. Let $U\subseteq V(\J)$ and $\MSpine_\J(U) = \{(\G_i,U_i)\}_{i=1}^s$.
    The set $U$ is \emph{$\J$-stable} if and only if
    for every $w\in V(\J)\setminus U$, there exists $i\in [s]$ such that $w \in V(\G_i)\setminus U_i$.
\end{definition}

\noindent
The following lemma implies that once our algorithm discovers a vertex $\min_\G V(\J)$, where $\J$ is some $H$-subgraph, the discovered vertices form a $\J$-stable set until it discovers all vertices of some (potentially different to~$\J$) $H$-subgraph.

\begin{lemma}\label{lem:HStable}
    Let~$\J$ be an $H$-subgraph in~$\G$ and let $U\subseteq V(\J)$ be $\J$-stable. Let $\MSpine (U) = \{(\G_i,U_i)\}_{i=1}^s$. Then
   \begin{enumerate}
       \item\label{item:AllofH} $\MSpine_\J(U) = \emptyset$ if and only if $U= V(\J)$,
       \item\label{item:PlusNadir} if $U\neq V(\J)$, then for every $i\in [s]$, the set $U\cup \Nadir(\G_i,U_i)$ is $\J$-stable, and
       \item\label{item:MinimumStable} the set $\{\min_\J V(\J)\}$ is $\J$-stable.
   \end{enumerate}
\end{lemma}
\begin{proof}
We begin with (\ref{item:MinimumStable}). 
Let $m = \min_\J V(\J)$, and $V_1,\dots,V_s$ be the sets of vertices of the connected components of the graph $\J[V(\J)\setminus \{m\}]$.
For every $i\in [s]$, let $\J_i = \J[V_i \cup \{m\}]$.
We note that for every $i\in [s]$, $(\J_i,\{m\})$ is an useful pair since it satisfies the conditions of Definition~\ref{def:Huseful}.
Thus, by the definition of \Spine, we have $\MSpine_\J(\{m\}) = \{(\J_i,\{m\})\}_{i=1}^s$.
Now, notice that by construction, $\bigcup_{i=1}^s V_i = V(\J)\setminus \{m\}$ and therefore, the set $\{m\}$ satisfies the conditions for being $\J$-stable. 

We now prove (\ref{item:AllofH}).
Suppose $U= V(\J)$ and let $(\G',U_{\G'})$ be an arbitrary non-empty useful pair with~$\G'$ an induced ordered subgraph of~$\J$. By Proposition~\ref{prop:NoEmptyUseful}, $U_{\G'}\subsetneq V(\G')$ and thus $V(\G') \cap U  = V(\G') \neq U_{\G'}$. So, by the definition of \Spine, $(\G',U_{\G'})\not\in \MSpine_\J(U)$. That is, $\MSpine_\J(U) = \emptyset$.
Suppose $U\neq V(\J)$ and let $w \in V(\J)\setminus U$.
By the definition of $\J$-stable, there exists $(\G',U_{\G'})\in \MSpine(U)$ such that $w \in V(\G')\setminus U_{\G'}$ and hence $\MSpine_\J(U) \neq \emptyset$.

Finally, we prove (\ref{item:PlusNadir}).
Let $k\in [s]$ and  $v = \Nadir(\G_k,U_k)$.
If $U \cup \{v\} = V(\J)$, then by (\ref{item:AllofH}), $\MSpine_\J(U\cup \{v\})=\MSpine_\J(V(\J))=\emptyset$, which by definition implies that $U \cup \{v\}$ is $\J$-stable. 
So, assume that $U \cup \{v\} \subsetneq V(\J)$.

By Lemma~\ref{lem:Spine}, for every distinct $i,j\in [s]$, we have $(V(\G_i)\setminus U_i) \cap (V(\G_j)\setminus U_j) = \emptyset$.
Thus, we may conclude that for every $i\in [s]\setminus\{k\}$, we have $(\G_i,U_i) \in \MSpine(U\cup \{v\})$.
Then in order to show that $U\cup \{v\}$ is $\J$-stable, we only need to show that for every $x \in V(\G_k)\setminus (U_k \cup \{v\})$, there exists
$(\G^*,U^*) \in \MSpine_\J(U \cup \{v\})$ such that $x \in V(\G^*)\setminus U^*$.

If $V(\G_k) = U_k\cup \{v\}$, then the above is trivially valid.
So, suppose that $U_k\cup \{v\} \subsetneq V(\G_k)$ and define $\J_k := \J[V(\G_k) \setminus(U_k \cup \{v\})]$. Suppose first that $\J_k$ is a connected subgraph.

By the construction of $\J_k$, there exists a minimal subset $U'_k \subseteq U_k\cup \{v\}$, such that every vertex in $U'_k$ is adjacent to a vertex in $V(\J_k)$.
We note that also, as $\J_k = \J[V(\G_k) \setminus(U_k \cup \{v\})]$, $U'_k \subseteq U_k\cup \{v\}$ and $v=\Nadir(\G_k,U_k)$, we have that for every $x\in U'_k$ and $y\in V(\J_k)$ it holds that $x <_\J y$.
Finally, by construction, every vertex in $V(\J_k)\cup U'_k$ that is connected to a vertex not in $V(\J_k)\cup U'_k$ is in $U'_k$, and hence
$U'_k$ is a prefix of $\J[V(\J_k)\cup U'_k]$. Since~$\J_k$ is connected,
$(\J[V(\J_k)\cup U'_k],U'_k)$ is a useful pair and since $U'_k\cup \{v\} \subseteq U\cup \{v\}$, we know that $(\J[V(\J_k)\cup U'_k],U'_k) \in \MSpine(U \cup \{v\})$.

If $\J_k$ is not connected then each of its connected components must have a vertex that is adjacent to a vertex in $\J_k$, otherwise we have a contradiction since $\J$ must be connected (since $H$ is connected).
With this in mind, the proof can be concluded by applying the same reasoning as we applied to $\J_k$ to each connected component of $\J_k$.
\end{proof}

\newcommand{\NadirNumbers}{2r^2\cdot p^{r^2}}
\newcommand{\NadirNumbersTT}{4r^2\cdot p^{r}}

\noindent
The following lemma is required for proving that our algorithm finds enough useful pairs that share only their prefix. This is crucial for applying the replacement idea we described in Step~(\ref{step:cover}) above.

\begin{lemma}~\label{lem:ParallelHSubgraphs}
\begin{sloppy}
   Let $(\MD,U)$ be a \Strata.
    If $|\Nadir(\MD,U)| \geq \NadirNumbers$, then there exists a \Strata~$(\MD',U)$ 
    such that $\MD' \subseteq \MD$, $|\MD'| = |\Nadir(\MD',U)| \geq 2r$
    and for every distinct $(\G_1,U),(\G_2,U)\in \MD'$, we have $V(\G_1)\cap V(\G_2) = U$. 
\end{sloppy}
\end{lemma}
\begin{proof}
    Let~$(\tilde{\mathcal M},U)$ be a \Strata~such that $\tilde{\mathcal M} \subseteq \MD$ and for every $x \in \Nadir(\MD,U)$, there exists exactly one useful pair $(\G_x,U) \in \tilde{\mathcal{M}}$
    such that $x = \Nadir(\G_x,U)$. By assumption, $|\tilde{\mathcal M}| \geq 2r \cdot \ell p^{r\ell}$.
    
    We first argue that every vertex~$v \in V(\G) \setminus U$ can be contained in at most~$p^{r\ell}$ members of~$\tilde{\mathcal M}$. To that end, assume that~$v$ is contained in members~$(\G_1, U),\ldots,(\G_t,U) \in \tilde{\mathcal M}$. 
    Recall that by the definition of a \Strata~and useful pairs, the graphs~$\G[V(\G_i)\setminus U]$ are connected and $\Nadir(\G_i, U)\leq_\G y$, for every vertex $y\in V(\G_i)\setminus U$. 
    Thus, for every $i\in [t]$, there exists an $(r,\G)$-\chain~$P_i$ whose vertices are in $V(\G_i)\setminus U$, with $v$ as its start vertex and $\Nadir(\G_i, U)$ as its end vertex, $\length(P_i) \leq \ell-1 \leq r$.
    Since, by the definition of \Strata, for every $i\in [t]$, the only edges common to any two graphs in useful pairs of $\tilde{\mathcal M}$ are between the vertices of $U$, we conclude that all the $(r,\G)$-\chain s $P_1,\ldots,P_t$ are edge-disjoint.
     
    Thus, by Lemma~\ref{lemma:chain-bundle}, $t \leq p^{r\ell}$ and hence $v$ can be contained in at most $p^{r\ell}$ members of~$\tilde {\mathcal{M}}$.
    Consequently, a simple greedy selection process can construct the claimed \Strata~$(\MD',U)$.
\end{proof}

\noindent
Finally, the following lemma is used to show that the trimming process in the next section does not remove too many edges from $G$ when constructing $\GT$.

\newcommand{\numOfStrataPerVertex}{p^{3r^2\log{r}}}
\begin{lemma}\label{lem:StrataPerVertexCounting}
    For every $v\in \G$, every set of pairwise non-similar useful pairs $(\G^*,U^*)$ such that $v = \max_\G U^*$ has a size of at most $\numOfStrataPerVertex$.
\end{lemma}
\begin{proof}
By definition, for every pair $(\G^*,U^*)$ such that $v = \max_\G U^*$ we have $U^*\setminus \{v\} \subseteq \Target^r_\G(v)$.
Since $\G$ is ordered, the number of options for the set $U$ is $\binom{|\Target^r_\G(v)|}{|V(H)|-1} \leq p^{2r|V(H)|}$, where the inequality follows because by Proposition~\ref{prop:SmallNeighbours}, $|\Target^r_\G(v)| \leq p(p-1)^r$.

The number of options for $\G^*$ is at most the number of options to choose an induced subgraph of $H$ times the number of different orders of $H$. This is at most $2^{2|V(H)|\log{|V(H)|}}$.
Thus, as $|V(H)|\leq r$, every set of pairwise non-similar useful pairs $(G^*,U^*)$ such that $v = \max_\G U^*$ has a size of at most $\numOfStrataPerVertex$.
\end{proof}

\section{Trimming}\marginnote{Trimming}\label{sec:Trimming}
In this section, we present the procedure that we refer to as \emph{trimming}.
We stress again that the trimming procedure is used only for the analysis of our algorithm; as such, we do not care about any aspect of the
complexity of the trimming. 

Trimming receives as input a graph $G$ with a $(p,r)$-admissible ordering $\G$,
where $p > 1$,
a $(p,r)$-admissible graph $H$ on at most $r$ vertices
and parameters $\alpha,\beta,\delta \in \Naturals$.
From now on, in this section, assume that these parameters are fixed.
Trimming creates a subgraph $\GT$ of $G$ that has a number of properties that are essential for the analysis of our algorithm.

\paragraph*{The trimming procedure}
$\GT$ is initially set to $\G$ and
the following steps iterate over all vertices $v\in V(\G)$ repeatedly, until none of the steps results in the removal of the edges from $\GT$:
\begin{enumerate}
    \item\label{step:Neighbours} If $\deg_\G(v)> \alpha$, then for every $u\in \Target^1_{\G}(v)$, the edge $vu$ is removed from $\GT$.
    \item\label{step:admissiblePaths} If for some $i\in [r]$ and $u\in \Target^i_{\G}(v)$,
     $\GT$ has a maximum set $\mathcal S$ of edge-disjoint, $i$-admissible paths $uPv$, each of length $i$ and $|\mathcal S|\leq\deg_\G(v)/\beta$, then all the edges of the paths of $\mathcal S$ are removed from $E(\GT)$.
    \item\label{step:TrimStrata} If there exists a \textbf{maximum} \Strata $(\MD,U)$ such that $v = \max_\G U$, and $|\MD| < \deg_\G(v)/\beta$, then all edges participating in $\MD$ are removed from $E(\GT)$.
    \item\label{step:TrimDeltaStrata} If there exists a maximum \dStrata{$\delta$} $(\MD,U)$ such that $v = \max_\G U$, and $|\MD| < \deg_\G(v)/\beta$,  then all the edges participating in $\MD$ are removed from $E(\GT)$.
\end{enumerate}

\noindent
We first show that if~$G$ is far from~$H$-freeness, then so is the graph $\GT$ created by the trimming process.

\newcommand{\betaLB}{8p^{16r^2\log{r}}/\epsilon'}
\newcommand{\alphaLB}{8p^2/\epsilon'}
\newcommand{\deltaVal}{2^{12}\beta^{3}r}

\begin{lemma}\label{lem:TrimDist}
Let $p,r \in \Naturals$, $\epsilon' > 0$, $\beta = \betaLB$, $\alpha \geq  \alphaLB$ and
$G$ be a $(p,r)$-admissible graph that is $2\epsilon'$-far from $H$-freeness.
The graph $\GT$ created by trimming $\G$ with parameters $\alpha$, $\beta$ and $\delta > 0$, is $\epsilon'$-far from $H$-freeness.
\end{lemma}
\begin{proof}
We show that in every stage of trimming, at most $\epsilon'|V(G)|/4$ edges are removed. 
Since there are $4$ stages in trimming, at most $\epsilon'|V(G)|$ edges are removed during trimming to obtain $\GT$ from $\G$.
Since $\G$ is $\epsilon'$-far from $H$-freeness, the above implies that $\GT$ is $\epsilon'$-far from $H$-freeness.

\textbf{Trimming step (\ref{step:Neighbours})}.
For every very $v\in V(\G)$ such that $\deg_{\G}(v) > \alpha$, at most $|\Target^{1}_{\G}(v)| \leq p$, where the inequality follows from Proposition~\ref{prop:SmallNeighbours}.
Let $W$ be the set of all vertices $v$ such that $\deg_{\G}(v) > \alpha$.
We note that $|W| \leq 2|E(G)|/\alpha \leq 2p|V(G)|/\alpha$, where the inequality follows from Fact~\ref{fact:num-edges}.
Thus, by the above, the total number of edges removed in (\ref{step:Neighbours}), is at most $|W|p \leq 2p^2|V(\G)|/\alpha \leq \epsilon'|V(\G)|/4$, where the inequality follows since $\alpha \geq \alphaLB$.

\textbf{Trimming step (\ref{step:admissiblePaths})}.
For every $v\in V(G)$, $i\in [r]$ and $u\in \Target^i_\G(v)$, the maximum number of edges removed in this step is $i\deg_\G(v)/\beta$.
Hence, the total number of edges removed in this step is at most
$$
    \sum_{v\in V(\G)}\sum_{i\in [r]}\sum_{u\in \Target^i_\G(v)} \!\!\! i\cdot\deg_\G(v)/\beta \leq
(r \cdot p(p-1)^r/\beta)\cdot\sum_{v\in V(\G)} \! \deg_\G(v) \leq \epsilon'|V(\G)|/4,
$$
where the first inequality follows since, by Proposition~\ref{prop:SmallNeighbours}, for every $i\in [r]$, $|\Target^i_\G(v)|\leq p(p-1)^r$ and the second because $\beta \geq \betaLB \geq 8r \cdot p(p-1)^r/\epsilon'$.

\textbf{Trimming step (\ref{step:TrimStrata})}.
By Lemma~\ref{lem:StrataPerVertexCounting}, for every $v\in \G$, every set of pairwise non-similar useful pairs $(G^*,U^*)$ such that $v = \max_\G U^*$ has a size of at most $\numOfStrataPerVertex$.
According to step (\ref{step:TrimStrata}) for each strata we removed at most $|E(H)|\deg_{\G}(v)/\beta$ edges.
Hence, the total number of edges removed in step (\ref{step:TrimStrata}),
is at most 
$$\sum_{v\in V(\G)}\numOfStrataPerVertex |E(H)|\deg_{\G}(v)/\beta \leq (\numOfStrataPerVertex/\beta)\sum_{v\in V(\G)}\deg_{\G}(v) \leq \epsilon'|V(\G)|/4,
$$
where the inequalities follow because $p \geq 2$, $|V(H)| \leq r$, $|E(H)| \leq r(r - 1)/2$, $\beta \geq \betaLB$, and by Fact~\ref{fact:num-edges}, $\sum_{v\in V(\G)}\deg_{\G}(v) \leq p|V(\G)|$. 

\textbf{Trimming step (\ref{step:TrimDeltaStrata})}.
The computation for this step is the same as the computation for step (\ref{step:TrimStrata}).

\end{proof}

\begin{proposition}\label{prop:Light}
    If $uv \in E(\GT)$, where $u>_{\GT} v$, then $\deg_\G(u) \leq \alpha$.
\end{proposition}
\begin{proof}
$\deg_\G(u) \leq \alpha$ must hold because otherwise, by trimming step (\ref{step:Neighbours}), $uv \not\in E(\GT)$, a contradiction.
\end{proof}


\section{The graph querying algorithm}\label{sec:PBFS}
\newcommand{\xiOneVal}{\lceil 20\alpha/\epsilon' \rceil}
\newcommand{\xiTwoVal}{r^2+3r+1}
\newcommand{\xiThreeVal}{\lceil 20r\delta \rceil}

\newcommand{\OuterLoop}{Outer Loop}
\newcommand{\MiddleLoop}{Middle Loop}
\newcommand{\InnerLoop}{Inner Loop}

We present next the algorithm for testing $H$-freeness in $(p,r)$-admissible graphs, where $|V(H)|\leq r$.
We assume that the parameters $r,p$ and $H$ are hard-coded in the algorithm.
We also assume that both the input graph~$G$ and the graph $H$ are $(p,r)$-admissible, since if~$H$ is not then the algorithm may reject, because every subgraph of a $(p,r)$-admissible graph is itself $(p,r)$-admissible.
In the algorithm, we use $\alpha$, $\beta$ and $\delta$ without providing their exact values. The values are provided in the final theorem as the end of this section.

\newcommand{\epsilonPVal}{\epsilon/2}

\smallskip
\begin{GrayBox}{\textbf{Algorithm~\ref{alg:PBFS}}}
   \begin{algorithm}[H]\label{alg:PBFS}
    \DontPrintSemicolon
    \SetNoFillComment
      
    \KwInput{$\epsilon >0$, random neighbour oracle access to a graph $G$ and $|V(G)|$} 
    Set $\epsilon'= \epsilonPVal$, $\alpha = \alphaLB$, $\beta = \betaLB$ and $\delta = \deltaVal$\;
    Set $\xi_1 = \xiOneVal$, $\xi_2 = \xiTwoVal$ and $\xi_3 =\xiThreeVal$\;\label{line:XiValues}
    Set $s = 0$\;
    Set $G_s$ to be the empty graph\;
    \RepTimes{$\xi_1$ \label{line:Initials}}{
            Add to $V(G_s)$ an independently and u.a.r selected vertex from $V(G)$\;
    }  
    \RepTimes(\tcp*[f]{\OuterLoop}){$\xi_2$ }{\label{line:Outer}
        Set $s = s + 1$\;
        Set $V(G_{s}) = V(G_{s-1})$, $V(E_{s}) = V(E_{s-1})$\;
        \For(\tcp*[f]{\MiddleLoop}){$v\in V(G_{s-1})$}{\label{line:Middle}
            \RepTimes(\tcp*[f]{\InnerLoop}){$\xi_3$ }{\label{line:Inner}
                Query the oracle for a random neighbour $u$ of $v$\;
                    $V(G_s) = V(G_s) \cup \{u\}$\;
                    $E(G_s) = E(G_s) \cup \{vu\}$\;
                }
        }
    }
    \If{if $G_{\xi_2}$ has an $H$-subgraph}{
        Reject\\}\Else{Accept}        
   \end{algorithm}
\end{GrayBox}

\noindent
We refer to the loop on line~\ref{line:Outer} of Algorithm~\ref{alg:PBFS}, as the \emph{\OuterLoop}, to the loop on line~\ref{line:Middle} as the \emph{Middle Loop}~ and to the loop on line~\ref{line:Inner} as the \emph{\InnerLoop}.
In all this section, $G$ with a subscript refers to the relevant value once after Algorithm~\ref{alg:PBFS} completed its run.
In all of this section, except for the final theorem, $G$ and $H$ are 
a $(p,r)$-admissible graphs, $|V(G)|= n$, $V(H) \leq r$,
 $\GT$ is a subgraph of $G$ obtained by trimming with the parameters $r,\alpha,\beta,\delta\in \Naturals$, and
\marginnote{Assumptions $\xi_3$}
when we use the parameter $\xi_3$ we assume that it is a multiple of $r$.

\begin{lemma}\label{lem:PBFSQC}
The query complexity of Algorithm~\ref{alg:PBFS} is independent of the size of the input graph $G$.
\end{lemma}
\begin{proof}
   \sloppy
   According to the code of Algorithm~\ref{alg:PBFS}, the query complexity of Algorithm~\ref{alg:PBFS} is $O(\xi_1\xi_3^{\xi_2})$.
   The values of $\xi_1$, $\xi_2$ and $\xi_3$ depend only on the $r$, $\epsilon'$, $\alpha$ and $\delta$.
   The values of $\epsilon'$, $\alpha$ and $\delta$  (with a bit of extra work) depend only on $p$, $r$, $\epsilon$ and $|V(H)|$.
\end{proof}



\newcommand{\UsefulSet}[1]{\tilde{#1}}

\newcommand{\probHitting}{1-(1-\epsilon'/\alpha)^{\xi_1}}

\noindent
The following lemmas are used by the last theorem of this section, which states that our algorithm is a tester for $H$-freeness.
Lemma~\ref{lem:ProbHittingHSub} is used to prove that, with a sufficiently high probability, $V(G_0)$ includes a vertex of an $H$-subgraph $\J$ of $\GT$.
Lemma~\ref{lem:MinimumVertex} is used to show that, given the event just described, with a
sufficiently high probability Algorithm~\ref{alg:PBFS} discovers $\min_{\GT} V(\J)$.
Lemma~\ref{lem:TargetVertex} is used to prove Lemma~\ref{lem:MinimumVertex}, \ref{lem:ProbableNadir} and~\ref{lem:ManyNadirs}.
Lemma~\ref{lem:ProbableNadir} and~\ref{lem:ManyNadirs} are used to show that given that the event described above occurs, then with sufficiently high probability, all the vertices of some $H$-subgraph of $\GT$ are discovered by Algorithm~\ref{alg:PBFS}.
Lemma~\ref{lem:AllEdges} is used to show that, given that the event described above occurs, with sufficiently high probability, the edges of the just-mentioned $H$-subgraph are discovered by Algorithm~\ref{alg:PBFS}.

\begin{lemma}\label{lem:ProbHittingHSub}
Let $\epsilon' > 0$.
If $\GT$ is $\epsilon'$-far from being $H$-free, 
then with probability at least $\probHitting$, $V(G_0)$ includes a vertex of an $H$-subgraph of $\GT$.
\end{lemma}
\begin{proof}
    Let $\J$ be an $H$-subgraph of $\GT$.
    Since $|\J| > 1$ and $\J$ is connected, by Proposition~\ref{prop:Light}, there exists a vertex $v\in V(\J)$ such that $\deg_\G(v) \leq \alpha$.
    Thus, we may conclude that every $H$-subgraph of $\GT$ includes a vertex $v$ such that $\deg_\G(v) \leq \alpha$.
    Let $D$ be the union of all these vertices.
    By construction, if we remove from $\GT$ every edge incident on a vertex of $D$, then we remove from $\GT$ at most $\alpha|D|$ edges, and the resulting graph is $H$-free.
    Since $\GT$ is $\epsilon'$-far from being $H$-free, we can conclude that $\alpha|D| \geq \epsilon' |V(G)|$ and therefore $|D| \geq \epsilon' |V(G)|/\alpha$.
    Finally, by line~\ref{line:Initials} of the Algorithm, with probability at least $\probHitting$,
    $V(G_0)$ includes a vertex of an $H$-subgraph of $\GT$.  
\end{proof}

\newcommand{\probTargetVertex}[1]{1 - #1(1 - \delta^{-1})^{\xi_3}}

\begin{lemma}\label{lem:TargetVertex}
Let $u,v \in V(G)$ be such that $\large{u>_{\GT} v}$.
Assume that there exists in $\GT$ an $(r,\GT)$-\chain~of length $\ell$ with $u$ as its start vertex and $v$ as its end vertex.
If $\alpha \leq \delta$ and $u\in V(G_\tau)$, where $\tau \leq \xi_2-\ell$, then with probability at least $\probTargetVertex{\ell}$, we have $v\in V(G_{\tau + \ell})$.
\end{lemma}

\begin{proof}
We prove the lemma by induction on the value of $\ell$.
Suppose that $\ell = 1$. 
Thus, $u$ and $v$ are adjacent in $\GT$.
As $u>_{\GT} v$, according to (\ref{step:Neighbours}) of trimming, $\deg_{\G}(u) \leq \alpha$.
Since $u\in V(G_\tau)$ due to the \InnerLoop, with probability at least $1-(1-\alpha^{-1})^{\xi_3} \geq 1 - \ell(1 - \delta^{-1})^{\xi_3},$ we have $v\in V(G_{\tau + 1})$.

Assume by induction that the lemma holds for every $(r,\GT)$-\chain~of length  at most $\ell-1$. 
Let $uLv$ be an $(r,\GT)$-\chain~of length $\ell$ in $\GT$.
Let $x$ be the closest vertex to $u$ in $uLv$ such that $x<_\G u$ and let $uL'x$ be a subpath of $uLv$.
By construction, for every $z\in L'$, we have $z >_{\GT} u >_{\GT} x$ and hence $uL'x$ is an $(r,\GT)$-admissible path and, by definition, also an $(r,\GT)$-\chain.

Suppose that $x\neq v$.
Let $\ell' = \length(uL'x)$.
Since $u\in V(G_\tau)$, by the induction assumption with probability at least
$1 - \ell'(1 - \delta^{-1})^{\xi_3}$, $u\in V(G_{\tau + \ell'})$.
Let $xL^*v$ be a subpath of $uLv$, by Proposition~\ref{prop:subChain},  $xL^*v$ is a $(r,\GT)$-\chain.
So, again by the induction assumption,
if $x\in V(G_{\tau + \ell'})$, then with probability at least
$1 - (\ell-\ell')(1 - \delta^{-1})^{\xi_3}$, we have $v\in V(G_{\tau + \ell})$.
Thus, we can conclude that the lemma holds when $x\neq v$.

Suppose that $x = v$.
Thus, by the above, $uLv$ is an $(r,\G)$-admissible path in $\GT$.
Therefore, by step~\ref{step:admissiblePaths}, $\mathbf \G$ has at least $\deg_\G(u)/\delta$
edge-disjoint $(r,\G)$-admissible paths each of length $k$, and with $u$ as the start vertex and $v$ as the end vertex.
Since these paths are edge-disjoint, at least $\deg_\G(u)/\delta$ of the neighbours of the vertex $u$ are in one of these paths.
Thus, according to the inner loop with probability at least $1-(1 - \delta^{-1})^{\xi_3}$,
if $u\in V(G_\tau)$, then $y\in V(G_{\tau+1})$ for one of the paths $uy\tilde{L}v$ just mentioned.

Let $uy\tilde{L}v$, one of the $(r,\GT)$-\chain s mentioned above.
By definition, as $uLv$ is an $(r,\GT)$-\chain, so is $uy\tilde{L}v$.
Thus, by the induction assumption, if $u\in V(G_{\tau + 1})$, then with probability at least $1-(\ell-1)(1 - \delta^{-1})^{\xi_3}$, it holds that $v\in V(G_{\tau + \ell})$. 
Thus, we can conclude that the lemma holds when $x = v$, and the proof is done.
\end{proof}

\newcommand{\probFindMin}{\probTargetVertex{r}}

\begin{lemma}\label{lem:MinimumVertex}
Let $\J$ be an $H$-subgraph of $\GT$ and $u\in V(\J)$.
If $\alpha \leq \delta$ and $u\in V(G_\tau)$, where $\tau \leq \xi_2-r$, then with probability at least $\probFindMin$, we have $\min_{\J}V(\J)\in V(G_{\tau + r})$. 
\end{lemma}
\begin{proof}
Let $v = \min_{\J}V(\J)$.
Since $|V(H)|\leq r$, and $\J$ is an $H$-subgraph of $\GT$,
there exists an $(r,\GT)$-chain in $\GT$ with $u$ as a start vertex and $\min_{\J}V(\J)$ as an end vertex.
Since also $\alpha \leq \delta$,
by Lemma~\ref{lem:TargetVertex}, if $u \in V(G_\tau)$, then
$\min_{\J}V(\J) \in V(G_{\tau + r})$, with probability at least
$\probFindMin$.
\end{proof}

\noindent
The following lemma appears quite similar to Lemma~\ref{lem:TargetVertex}, so let us point out the important difference first:
Lemma~\ref{lem:TargetVertex} deals with the probability that Algorithm~\ref{alg:PBFS} discovers a vertex $v$ given that it already discovered a vertex $\mathbf{u>_{\GT} v}$ and there exists an $(r,\GT)$-\chain~$(uPv)$;
the following lemma deals with the probability that Algorithm~\ref{alg:PBFS} discovers a vertex $u$, given that it already discovered a vertex $\mathbf{v>_{\GT} u}$ that satisfies some extra condition.
We will need it to analyze the probability of discovering a $\Nadir$ vertex.

\begin{lemma}\label{lem:ProbableNadir}
Let $u,v \in V(G)$ be such that $\mathbf{v>_{\GT} u}$.
Assume that there exists a family $\mathcal{S}$, of edge-disjoint $(r,\GT)$-admissible paths, which have $u$ as the start vertex and $v$ as the end vertex, and that $|\mathcal{S}| \geq \deg_\G(v)/\delta$.
If $\alpha \leq \delta$ and $v \in G_\tau$, where $\tau \leq \xi_2-\ell$, then with probability at least $1-r(1 - \delta^{-1})^{\xi_3},$
 we have $u \in G_{\tau + r}$.
\end{lemma}
\begin{proof}
 Since the paths in $\mathcal{S}$ are edge-disjoint, at least $\deg_\G(v)/\delta$ of the neighbours of the vertex $v$ are on one of these paths.
Thus, according to \InnerLoop, with probability at least $1-(1 - \delta^{-1})^{\xi_3}$,
if $v$ is in the knowledge graph at the end of iteration $\tau$, then a vertex $y$ of a path in $\mathcal{S}$ is in the knowledge graph at the end of iteration $\tau+1$. 

Let $yLu$ be the subpath of the path in $\mathcal{S}$ that includes $y$.
Since $yLu$ is a subpath of an $(r,\GT)$-admissible path, by Proposition~\ref{prop:subChain} $yLu$ is an $(r,\GT)$-\chain.
Thus, by Lemma~\ref{lem:TargetVertex}, given that $y\in G_{\tau+1}$, with probability at least $1-(r-1)(1 - \delta^{-1})^{\xi_3}$, we have $u\in G_{\tau + r}$.
This together with the previous implies the lemma.
\end{proof}

\newcommand{\manyNadirsProb}{1 - 1/(20r)}
\newcommand{\manyNadirsProbI}{1 - 1/(40r)}

\begin{lemma}\label{lem:ManyNadirs}
Let $(\G',U_\G')$ be a useful pair, and
 $(\MD,U)$ be a \dStrata{$\delta$} in $\GT$.
If $(r-1)(1 - \delta^{-1})^{\xi_3} < 1/2$, $\delta/64\beta^2 > 16\ln r$,$\xi_3 > \delta/(2\beta)$ and
$\max_{\G} U \in G_{\tau}$~then,
with probability at least 
$\manyNadirsProb$,
  the graph $G_{\tau+r}$ includes at least $\delta/(16\beta^2)$ vertices from $\Nadir(\MD,U)$.
\end{lemma}
\newcommand{\col}{\mbox{col}}

\begin{proof}
Let $v = \max_\G U$ and suppose that, $(r-1)(1 - \delta^{-1})^{\xi_3} < 1/2$, $\delta/64\beta^2 > 16\ln r$, $\xi_3 > \delta/(2\beta)$ and $v \in G_{\tau}$ where $\tau \leq \xi_2 - r$ of \OuterLoop.
Let $S = N_{G}(v) \cap V(G_{\tau+1})$, 
and $W=(S,\Nadir(\MD,U),F)$ be a bipartite graph such that
$xy \in F$ if and only if $x\in S$, $y \in \Nadir(\MD,U)$ and $\GT$ has an $(r,\GT)$-\chain~$xLy$.
Also, let $Y$ be the size of a maximum matching in $W$.  

We first show that if $Y\geq \delta/(4\beta^2)$ at the end of iteration $\tau + 1$ of \OuterLoop, then with probability at least $\manyNadirsProbI$, at the end of iteration $\tau + r$ of \OuterLoop,
$G_{\tau+r}$ includes at least $\delta/(16\beta^2)$ vertices from $\Nadir(\MD',U)$.
Afterwards, we show that with probability at least $\manyNadirsProbI$, at the end of iteration $\tau + 1$ of \OuterLoop, $Y \geq \delta/(4\beta^2)$.
This together with the previous implies the lemma.

Suppose that $Y \geq \delta/(4\beta^2)$, and let $M\subseteq F$ be a maximum matching in $W$.
For every edge $xy\in M$, where $x\in S$ and $y \in \Nadir(\MD,U)$,
the graph $\GT$ has an $(r,\GT)$-\chain~$xLy$.
Thus, by Lemma~\ref{lem:TargetVertex}, for every edge $xy\in M$, it holds that $y \in G_{\tau + r}$, with probability at least $1 - (r-1)(1 - \delta^{-1})^{\xi_3} \geq 1/2$, where the inequality follows since $(r-1)(1 - \delta^{-1})^{\xi_3} < 1/2$.
By definition of Algorithm~\ref{alg:PBFS} (see \MiddleLoop~condition), these events are independent.
Therefore, by the Chernoff bound, the graph $G_{\tau+r}$ includes at least $\delta/(16\beta^2)$ vertices of $\Nadir(\MD,U)$, with probability at least $1 - e^{-\delta/64\beta^2} > \manyNadirsProbI$, where the inequality follows since $\delta/64\beta^2 > 16\ln r$. We next prove the second claimed probability.

\newcommand{\Exp}{\mathbb{E}}

Now, we show that, with probability at least $\manyNadirsProbI$, at the end of iteration $\tau + 1$ of \OuterLoop, $Y \geq \delta/(4\beta^2)$. We use martingales for the proof.
For every $i\in [\xi_3]$, let $X_i$ be the identity of the neighbour of $v$ selected in iteration $i$ of \InnerLoop, and let $W_i = (S_i,\Nadir(\MD',U),F_i)$ be a bipartite graph such that 
$S_i = \{X_j\}_{j=1}^{i}$  and
$xy \in F_i$ if and only if $x\in S_i$, $y \in \Nadir(\MD,U)$ and $\GT$ has an $(r,\GT)$-\chain~$xLy$.
We set $W = W_{\xi_3}$. 

For every $i \in [2,\xi_3]$, let $Y_i$ be the size of a maximum matching in $W_i$ minus the size of a maximum matching in $W_{i-1}$.
Since $\xi_3 > \delta/(2\beta)$ we have that the expectation $\Exp[Y]\geq  \sum_{i=1}^{\delta/(2\beta)}\Exp[Y_i]$ (where $\Exp$ denotes the expected value). 
Next, we show that for every $i\in [\delta/(2\beta)]$, we have $\Exp[Y_i] \geq 1/(2\beta)$.
Together, with the above, this implies that $\Exp[Y]\geq \delta/(4\beta^2)$.

By (\ref{step:TrimDeltaStrata}) of trimming  the definition of a \dStrata{$\delta$}, $|\MD| \geq \deg_\G(v)/\beta$.
So, if we remove from $(\MD,U,)$, all the useful pairs
$(\G^*,U)\in \MD$ such that $\Nadir(\G^*,U)$ is in a  subset of $\Nadir(\MD,U)$ of size at most $\delta/(2\beta)$, then the remaining \Strata~will be of size $\deg_\G(\max_\G U)/2\beta$.
Hence, the probability that a vertex selected in \InnerLoop~will be in one of the $\G$-useful pairs of $(\MD^*,U)$ is at least $1/(2\beta)$. Note that if this event occurs for some $i\in [2,\delta/(2\beta)]$, then $Y_i - Y_{i-1} = 1$.
Thus, we conclude that, indeed, for every $i\in [\delta/(2\beta)]$, we have $\Exp[Y_i] \geq 1/(2\beta)$.

We use the vertex exposure martingale to show that with sufficiently high probability $Y$ is sufficiently large. 
The whole graph $W$ is revealed by iterating over $i= 1,\dots,\xi_3$, and for every $i \in [\xi_3]$  the value $x$ of $X_i$ is revealed and the identity of the neighbours of $x$ in $W$ and the edges of $W$ incident on $x$.
For every $t\in [\xi_3]$, we let $Z_t = \Exp(Y\mid X_t,\dots,X_1)$.

We note that for every $t\in [\delta/(2\beta)]$, we have $|Z_{t+1} - Z_t| \leq 1$ since $Z_{t+1}$ has at most one fixed vertex more than $Z_t$.
Thus, by the Azuma--Hoeffding tail bound, we have
$|Z_{\delta/(2\beta)} - Z_0| \leq \delta/(8\beta^2)$,
with probability at least $1 - e^{-\frac{(\delta/(8\beta^2))^2}{2\delta/(2\beta)}} = 1 - e^{-\frac{\delta}{64\beta^3}} \geq \manyNadirsProbI$, where the last inequality follows because $\delta/64\beta^2 > 16\ln r$.
This implies that with probability at least $\manyNadirsProbI$, we have $Y_{\delta/(2\beta)} \geq \delta/(8\beta^2)$.
Thus, as $\xi_3 > \delta/(2\beta)$, we have $Y \geq Y_{\delta/(2\beta)}$, with probability at least $\manyNadirsProbI$.
\end{proof}

\begin{lemma}\label{lem:AllEdges}
Let $\tau \leq \xi_2-r$.
If $G_{\tau}$ includes all the vertices of an $H$-subgraph of $\GT$, then, with probability at least $1 - r^2(1-\alpha^{-1})^\delta$, the graph $G_{\tau + r}$ includes all the edges of the $H$-subgraph.
\end{lemma}
\begin{proof}
Let $\J$ be an $H$-subgraph in $\GT$ and suppose that $V(\J) \subseteq G_{\tau}$.
Let $uv \in E(\J)$, where $u>_\G v$.
By Proposition~\ref{prop:Light}, $\deg_\G(u)\leq \alpha$.
Hence, according to the \InnerLoop, with probability at least $1-(1-\alpha^{-1})^\delta$,
$uv\in G_{\tau + r}$.

For every vertex $u$ as above, at most $|V(H)|-1$ edges need to be discovered, by the union bound, the probability that this happens after $|V(H)|-1$ iteration of \OuterLoop, is at least $1 - (|V(H)|-1)(1-\alpha^{-1})^\delta$.
Since there may be up to $|V(H)|$, by the union bound, with probability at least $1 - |V(H)|(|V(H)|-1)(1-\alpha^{-1})^\delta \geq 1 - r^2(1-\alpha^{-1})^\delta$, the graph $G_{\tau + |V(H)|-1}$ includes all the edges of $\J$, and as $|V(H)|\leq r$ so does $G_{\tau + r}$.
\end{proof}

\begin{theorem}\label{thm:main}
If the input graph $G$ is $H$-free, then Algorithm~\ref{alg:PBFS} accepts with probability $1$. If the input graph $G$ is $(p,r)$-admissible and $\epsilon$-far from $H$-freeness, then Algorithm~\ref{alg:PBFS} rejects with probability at least $2/3$.
The query complexity of Algorithm~\ref{alg:PBFS} is independent of the size of the input graph. 
\end{theorem}
\begin{proof}

The query complexity of the statement follows from Lemma~\ref{lem:PBFSQC}.

Suppose first that $G$ is $H$-free.
Thus, as $G_{\xi_2}$ is a subgraph of $G$, $G_{\xi_2}$ is also $H$-free and 
hence Algorithm~\ref{alg:PBFS} will not discover an $H$-subgraph in $G_{\xi_2}$ and consequently accepts.

From here on, assume that $G$ is $\epsilon$-far from being $H$-free.
Let $\G=(G,\leq)$ be a $(p,r)$-order of the vertices of $G$. $\G$ exists since $G$ is $(p,r)$-admissible.
Let $\epsilon' = \epsilon/2$, $\GT$ be the subgraph of $G$ obtained by applying trimming to $\G$ with $H$ and parameters $r$, $\alpha = \alphaLB$, $\beta = \betaLB$, and $\delta = \deltaVal$. 
By Lemma~\ref{lem:TrimDist}, $\GT$ is $\epsilon'$-far from being $H$-free.

By Lemma~\ref{lem:ProbHittingHSub}, with probability at least $\probHitting\geq 19/20$, $V(G_0)$ includes a vertex of an $H$-subgraph of $\GT$, where the inequality follows because $\xi_1 = \xiOneVal$.
Let $\J$ be an $H$-subgraph of $\G$, such that $V(G_0)\cap V(\J) \neq \emptyset$.
Since $\alpha\leq \delta = \deltaVal$,
by Lemma~\ref{lem:MinimumVertex},  with probability at least 
$\probFindMin \geq 19/20$,
 it holds that $\min_{\J}V(\J)\in V(G_{r+1})$, where the inequality holds because $\xi_3 = \xiThreeVal  \geq 20r\delta$.

Let $v_1 = \min_{\J}V(\J)$.
By Lemma~\ref{lem:HStable}, $\{v_1\}$ is $\J$-stable.
Let $(\J_2,\{\phi(v_1)\})$ be a useful pair in $\MSpine_\J(\{v_1\})$,
by Lemma~\ref{lem:HStable} such a useful pair exists.

Let $v_2 = \Nadir(\J_2,\{v_1\})$.
Since $\J_2$ is a subgraph of $\GT$, by (\ref{step:TrimStrata}) and (\ref{step:TrimDeltaStrata})
it must be the case that there exists a \Strata~$(\MD_2,\{v_1\})$ in $\GT$ 
such that $|\MD_2| \geq \deg_{\G}(v_1)/\beta$ and all the useful pairs in $\MD_2$ are similar to $\J_2$.
We further on deal with two cases: in the first $(\MD_2,\{v_1\})$ is not a $\dStrata{\delta}$ in the second it is.

We note first that this is sufficient to prove the theorem because of the following.
In both the above mentioned cases, the result is that with probability at least $1-1/(20r)$,
$G_{2r+1}$ includes vertices $v_1,v_2$ that such that for some $\J'$ that is an $H$-subgraph of $\GT$ the set $\{v_1,v_2\}$ is $\J'$-stable.
Lemma~\ref{lem:HStable} ensures that
the same analysis can be repeated to ensure that with probability at least $19/20$, the graph $G_{r(r+1)+1}$ includes all the vertices of some $H$-subgraph of $\GT$.
Given that the previous occurred, by Lemma~\ref{lem:AllEdges}, with probability at least $1 - r^2(1-\alpha^{-1})^\delta > 19/20$ (where the inequality follows from the values of $\alpha$ and~$\delta$), the graph $G_{r(r+2)+1}$ will have an $H$-subgraph.
This together with all the above implies the theorem.

Suppose first that $(\MD_2,\{v_1\})$ is not a $\dStrata{\delta}$.
Then, by the definition of a $\dStrata{\delta}$, there exists $v_2 \in \Nadir(\MD_2,\{v_1\})$ for which there exists a \Strata~$(\MD_2',\{v_1\})\subseteq (\MD_2,\{v_1\})$, such that $|\MD_2'| >\deg_{\G}(v_1)/\beta$ and $\Nadir(\MD_2,\{v_1\}) = \{v_2\}$.
So, because $|V(H)| \leq r$ and $v_1 = \min_{\J}V(\J)$, this implies that
the number of $r$-\chain s of length at most $r$, starting in $v_1$ and ending in $v_2$ of length at most $r$ is at least $\deg_{\G}(v_1)/\delta$.
Since $\alpha \leq \delta$ and $v_1 \in G_{r+1}$, by Lemma~\ref{lem:ProbableNadir}, probability at least 
$1-r(1 - \delta^{-1})^{\xi_3} \geq 1-1/(20r)$,
 it holds that $v_2 \in G_{2r+1}$, where the inequality holds because $\xi_3 = \xiThreeVal  \geq 20r\delta$.
Since $v_2 = \Nadir(\J_2,\{v_1\})$) $(\J_2,\{\phi(v_1)\})$ be a useful pair in $\MSpine_\J(\{v_1\})$, by Lemma~\ref{lem:HStable}, the set $\{\phi(v_1),\phi(v_2)\}$ is $\J$-stable.

Suppose now that $(\MD_2,\{v_1\})$ is a $\dStrata{\delta}$.
Since $\xi_3 = \xiThreeVal$ and $\delta = \deltaVal$, we have that
$(2r-1)(1 - \delta^{-1})^{\xi_3} < 1/2$, $\xi_3 \geq \delta/(2\beta)$, $\delta/(64\beta^3) > 16\ln(r)$, and
$\max_{\G} \{v_1\} \in G_{\tau}$.
Thus,
by Lemma~\ref{lem:ManyNadirs}, with probability at least $1-1/(20r)$, the graph $G_{2r + 1}$ includes a set $W$ of at least $\delta/(16\beta^2)$ vertices from $\Nadir(\MD_2,\{v_1\})$.

Let $(\hat \MD_2,\{v_1\})$ be the \Strata such that
$\hat \MD_2 \subseteq \MD_2$ and  $W = \Nadir(\hat \MD_2,\{v_1\})$.
Since $\beta = \betaLB$ and $\delta = \deltaVal$,
it holds that $\delta/(16\beta^2) \geq \NadirNumbers$ and hence, by Lemma~\ref{lem:ParallelHSubgraphs}, there exists a \Strata $(\tilde \MD_2,\{v_1\})$, such that $\tilde \MD_2 \subseteq \hat \MD_2$, $|\tilde \MD_2| \geq 2r$,
and the set of vertices common to the graphs of every pair of distinct pairs in $\tilde \MD_2$ is exactly $\{v_1\}$.
Consequently, there exists $(\J^*,\{v_1\}) \in \tilde \MD_2$ 
such that $V(\J^*)\cap V(\J)= \{v_1\}$.
Let $v_3 = \min_{\J^*}V(\J^*) = \Nadir(\J^*,\{v_1\})$.

Using $\J^*$ and $\J$ we can construct a new $H$-subgraph $\J'$ of $\GT$ such that 
$\{v_1,v_3\}$ is an $\J'$-stable set (recall that $W\subset V(G_{2r + 1})$ and hence $v_3\in V(G_{2r + 1})$). 
The construction is done simply by removing the set of vertices and edges $\J_2$ (excluding $v_1$) from $\J$ and adding instead the vertices and edges of $V(\J^*)$.
It follows that $\J'$ is an $H$-subgraph of $\GT$, because $V(\J^*)\cap V(\J)= \{v_1\}$, $\J^*$ and $\J_2$ are order-isomorphic (both in the same \Strata) and $\J$ and $\J^*$ are ordered subgraphs of $\GT$.
It follows that $\{v_1,v_3\}$ is $\J'$-stable because of the following.
$\{v_1\}$ is an $\J'$-stable set because, by construction, $v_1 = \min_{J'}V(J')$.
Also by construction $(J^*,\{v_1\}) \in \MSpine_{\J'}(\{v_1\})$ and 
 $v_3 = \min_{J^*}V(J^*)$, so by Lemma~\ref{lem:HStable}, $\{v_1,v_3\}$ is an $\J'$-stable set.
\end{proof}

\noindent
Theorems~\ref{thm:testable-F} and~\ref{thm:testable-H} now follow immediately from Theorem~\ref{thm:main}.

\bibliographystyle{abbrv}
\bibliography{biblio}


\appendix

\section{Appendix}\label{sec:appx1}

  \ab* 
  This proof is taken from~\cite{H-freeness}.
\begin{proof}

    "Let $v$ be an arbitrary vertex in $V(G)$, $r \in \Naturals$, and $h\in [r]$.
    Note first that, by construction, $N_\G^-(v) = \Target^1_{\G}(v)$, and hence we only need to prove the bound on $|\Target^h_{\G}(v)|$.
        
    For every $u\in \Target^h_{\G}(v)$, let $vP_uu$ be a $r$-admissible path of length $h$; such a path exists by the definition of $\Target^h_{\G}(v)$.
    Let $W$ be a subgraph of $G$ defined as follows: $V(W)$ includes exactly the vertex $v$, all the vertices in $\Target^h_{\G}(v)$ and all the vertices in $P_u$, for every $u \in \Target^h_{\G}(v)$; and $E(W)$ includes every edge of $G$ participating in a path $vP_uu$ for some $u \in \Target^h_{\G}(v)$.
    
    By construction, the set of vertices $\Target^h_{\G}(v)$ is independent in $W$ and for every $w\in \Target^h_{\G}(v)$, $\deg_W(w) = 1$.
    Also, by construction, $\dist_W(v,\Target^h_{\G}(v)) \leq h$.
    Hence, we can find a rooted tree $T$ in $W$ with the following properties: $v$ is the root of $T$, the set $\Target^h_{\G}(v)$ is the set of leaves of $T$ and the depth of $T$ is at most $h$. 
    
    Next, we show that the degree of each vertex in the tree is at most $p$. This implies that indeed $|\Target^h_{\G}(v)| \leq p(p-1)^{h-1}$ and in particular $|N_\G^-(v)| \leq p$.
    
    Let $p'$ be the degree of $v$ in $W$.
    By the construction of $T$, there are $p'$ paths from $v$ to a leaf of $T$ (vertex in $\Target^h_{\G}(v)$).
    These paths have a maximum length $h$, share only the vertex $v$, and correspond to $r$-admissible paths of $\G$.
    The last part is valid since $u>_\G v$, for every internal vertex $u$ (non-leaf or root vertex) of $T$ and for every $w\in \Target^h_{\G}(v)$ (the leaves of $T$), we have $w<_\G v$.
    Therefore, these paths are an $r$-admissible packing of $\G$, and therefore their number $p' \leq p$, and in particular the degree of $v$ in $T$ is at most $p$.
    
    We now show that this also applies to every internal vertex $y$ of $T$.
    We also notice that $w<_\G y$, for every $w\in \Target^h_{\G}(v)$.
    However, it is not necessarily the case that $u<_\G y$ when $u$ is an internal vertex of $T$.
    This is resolved by noticing that for any $x\in \Target^h_{\G}(v)$ and path $yPx$ in $T$, if $P$ has a vertex smaller than $y$, then instead of taking the path $yPx$, we take its shortest subpath $yPx'$ such that $x' <_\G y$.
    Now, the same reasoning we used for $v$ implies that the degree of $y$ in $T$ is at most $p$. " 
\end{proof}

\end{document}